\documentclass[journal]{IEEEtran}
\usepackage{blindtext}
\usepackage{algorithm,algorithmic}
\usepackage{graphicx}
\usepackage{subfig}
\usepackage{amssymb,amsmath,mathtools,amsthm }
\usepackage[colorlinks]{hyperref}
\usepackage{url}
\usepackage{comment}
\usepackage{bm}
\usepackage{orcidlink}
\newcommand{\mat}[1]{{\boldsymbol{#1}}}

\newcommand{\ten}[1]{{\mathcal #1}}
\newcommand{\tY}{\ten{Y}}
\newcommand{\tX}{\ten{X}}
\newcommand{\tB}{\ten{B}}
\newcommand{\tT}{\ten{T}}
\newcommand{\tM}{\ten{M}}

\newcommand{\ba}{\boldsymbol a}

\newcommand{\bw}{\boldsymbol w} 
\newcommand{\bx}{\boldsymbol x} 
\newcommand{\bv}{\boldsymbol v} 
\newcommand{\bu}{\boldsymbol u} 
\newcommand{\by}{\boldsymbol y} 
\newcommand{\bb}{\boldsymbol b} 
 
\newcommand{\bc}{\boldsymbol c} 

\newcommand{\bone}{\boldsymbol 1} 

\newcommand{\blambda}{\boldsymbol \lambda} 
 
\newcommand{\bgamma}{\boldsymbol \gamma}

\newtheorem{theorem}{Theorem}
\newtheorem{lemma}{Lemma}

\newcommand{\vecc}{\text{vec} }
\newcommand{\diag}{\text{diag} }
\newcommand{\mbbE}{\mathbb E} 
\newcommand{\mbbR}{\mathbb R}
\newcommand{\mbbRp}{\mathbb R_{>0}}
\newcommand{\mbbRnn}{\mathbb R_{\geq 0}}
\newcommand{\mbbN}{\mathbb N_0} 

\newcommand{\Poi}{\text{ Poisson} }

\newcommand\sdots{\hbox to 1em{.\hss.\hss.}}

\DeclareMathOperator*{\argmax}{arg\,max}

\newcommand{\transpose}{^{\top}}

\newcommand{\mi}[1]{{\boldsymbol #1}}

\newcommand{\mVq}[1]{\mat{V}^{(#1)}}
\newcommand{\mtVq}[1]{\widetilde{\mat{V}}^{(#1)}}
\newcommand{\mhVq}[1]{\widehat{\mat{V}}^{(#1)}}
\newcommand{\mUp}[1]{\mat{U}^{(#1)}}
\newcommand{\mtUp}[1]{\widetilde{\mat{U}}^{(#1)}}
\newcommand{\mhUp}[1]{\widehat{\mat{U}}^{(#1)}}
\newcommand{\Omq}[1]{\mat{\Omega}^{(#1)}}
\newcommand{\Psip}[1]{\mat{\Psi}^{(#1)}}

\newcommand{\syi}{y_{(i)}}
\newcommand{\byi}{{\boldsymbol y}_{(i)}}

\newcommand{\tYi}{\tY_{(i)}}

\newcommand{\tYim}{\tY_{(i)\mi{m}}}

\newcommand{\byt}{{\boldsymbol y}_{(t)}}

\newcommand{\tYt}{\tY_{(t)}}

\newcommand{\tYtm}{\tY_{(t)\mi{m}}}

\newcommand{\tYtmo}{\tY_{(t-1)}}

\newcommand{\tYtmon}{\tY_{(t-1)\mi{n}}}

\newcommand{\bxi}{{\boldsymbol x}_{(i)}}

\newcommand{\tXi}{\tX_{(i)}}
\newcommand{\tXin}{\tX_{(i)\mi{n}}}

\newcommand{\sxt}{x_{(t)}}
\newcommand{\bxt}{{\boldsymbol x}_{(t)}}

\newcommand{\tXt}{\tX_{(t)}}
\newcommand{\tXtn}{\tX_{(t)\mi{n}}}

\newcommand{\thB}{\widehat{\tB}}
\newcommand{\tBnm}{\tB_{\mi{n},{\mi{m}}}}

\newcommand{\tYm}{\tY_\mi{m}}
\newcommand{\tBm}{\tB_\mi{m}}

\newcommand{\XiB}{\langle \tXi | \tB \rangle}

\newcommand{\XiBm}{\langle \tXi | \tB \rangle_\mi{m}}

\let\citep\cite
\let\citet\cite

\usepackage{xcolor}
\definecolor{blueee}{HTML}{000000}
\definecolor{greeen}{HTML}{000000}

\definecolor{orcidlogocol}{HTML}{A6CE39} 
\newcommand{\orcid}[1]{{\href{orcid.org}{\includegraphics[width=1.2em]{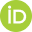}}}}

\title{Poisson-response Tensor-on-Tensor Regression and Applications}

\author{
Carlos~Llosa-Vite\orcidlink{0000-0001-7964-4201}
and 
Daniel~M.~Dunlavy\orcidlink{0009-0003-2299-4798}
\thanks{C. Llosa-Vite and D. Dunlavy are with Sandia National Laboratories, Albuquerque, NM, USA.}
\thanks{
Sandia National Laboratories is a multimission laboratory managed and operated by National Technology \& Engineering Solutions of Sandia, LLC, a wholly owned subsidiary of Honeywell International Inc., for the U.S. Department of Energy’s National Nuclear Security Administration under contract DE-NA0003525. SAND2026-19572R}
}

\begin{document}

\IEEEcompsoctitleabstractindextext{
\begin{abstract}
We introduce Poisson-response tensor-on-tensor regression (PToTR), a novel regression framework designed to handle tensor responses composed element-wise of random Poisson-distributed counts. Tensors, or multi-dimensional arrays, composed of counts are common data in fields such as international relations, social networks, epidemiology, and medical imaging, where events occur across multiple dimensions like time, location, and dyads. PToTR accommodates such tensor responses alongside tensor covariates, providing a versatile tool for multi-dimensional data analysis. We propose algorithms for maximum likelihood estimation under a canonical polyadic (CP) structure on the regression coefficient tensor that satisfy the positivity of Poisson parameters and then provide an initial theoretical error analysis for PToTR estimators. We also demonstrate the utility of PToTR through three concrete applications: longitudinal data analysis of the Integrated Crisis Early Warning System database, positron emission tomography (PET) image reconstruction, and change-point detection of communication patterns in longitudinal dyadic data. These applications highlight the versatility of PToTR in addressing complex, structured count data across various domains.
\end{abstract}
\begin{IEEEkeywords}
Poisson tensor decomposition,
Poisson regression,
tensor regression,
dyadic data,
sub-exponential random variables,
tensor decompositions,
positron emission tomography.
\end{IEEEkeywords}

}

\maketitle
\IEEEdisplaynontitleabstractindextext

\begin{figure}[b!]
    \centering       
    \vspace{-20pt}
    \includegraphics[width=\columnwidth]{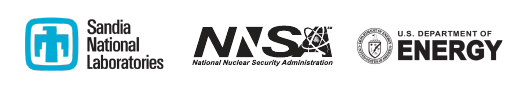}
    \vspace{-20pt}
\end{figure}

\section{Introduction}\label{sec:intro}
The classic identity-link Poisson regression model   \citep{mccullagh89,marschner10}
\begin{equation}\label{eq:poisreg}
   \syi \overset{\text{indep.}}{\sim}\Poi(\bb\transpose\bxi)
\end{equation}
relates the expected count response $\mbbE(\syi)$ as a linear combination of the multivariate covariate $\bxi$,  allowing for direct interpretation of the regression coefficient $\bb$. Throughout this paper, observations are enumerated using $i=1,2,\dots,I$ (or using $t=1,2,\dots,T$ when denoting time), scalars are lowercase (e.g., $y$), vectors are bold lowercase (e.g., $\by$), matrices are bold uppercase (e.g., $\mat{Y}$), and tensors, or multi-dimensional arrays, are script uppercase (e.g., $\tY$). Furthermore, $\mbbN$ denotes the set of natural numbers including $0$, $\mbbR$ the set of real values, $\mbbRp$ the set of positive real values, and $\mbbRnn$ the set of non-negative real values.  

When the response of an identity-link Poisson regression model $\byi$ is multivariate, the model can be formulated as
\begin{equation}\label{eq:mpoisreg}
\byi\overset{\text{indep.}}{\sim}\Poi(\mat{B}\bxi),
\end{equation}
which is short-hand for having the entries of $\byi$ independently follow a Poisson distribution with corresponding rates in $\mat{B}\bxi$. The number of entries in $\mat{B}$ grows considerably faster than the dimensions of $\byi$ or $\bxi$, making model fitting and inference intractable for most real-world applications. One could alleviate the issue of large numbers of parameters through $\ell_1$ lasso regularization \citep{tibshirani96}, or low-rank regularization on $\mat{B}$ \citep{mukherjeeandzhu11,fitzgeraldetal22}. However, such methodologies are limited to scalar- or vector-valued observations. In this work, we focus on tensor-on-tensor regression (ToTR)---where responses and covariates can be tensors, or multi-dimensional arrays (including scalars, vectors, matrices, and higher-order arrays)---for applications involving Poisson-distributed count data. In the next section, we introduce several motivating applications where such regression problems arise.

\subsection{Motivating application problems}\label{sec:motivation}

\subsubsection{Longitudinal data prediction}\label{sec:motiv_longitudinal}
Predicting relationships and actions using historical data is crucial in many longitudinal data analysis applications. For example, political analysts and strategists use the Integrated Crisis Early Warning System (ICEWS) database \citep{brien12} to inform policy and strategy associated with future international relations, which can have significant impact in global stability, security, and diplomacy \citep{zhao21}. By accurately forecasting interactions between countries using ICEWS data, governments, international organizations, and policymakers can take proactive measures to mitigate risks, allocate resources effectively, and make informed decisions. 

In previous work, Hoff modeled a subset of the ICEWS database consisting of weekly counts of directed relationships or actions taken between 25 countries and four classes of actions as a ToTR problem~\citet{hoff15}. The number of actions taken during the $t$-th week was encoded in a tensor $\tYt$ of size $25\times 25\times 4$ (country $\times$ country $\times$ action). Then, a simple order-1 autoregressive model for forecasting future actions can be formulated from Equation \eqref{eq:mpoisreg} with $\byt=\vecc(\tYt)$ as a vector response (where $\vecc(\cdot)$ denotes the tensor vectorization operation), and $\bxt=[1 \; \vecc(\tYtmo)\transpose]\transpose$ as a vector covariate. Due to the identity link, fixing the first entry of $\bxt$ to one ensures an additive intercept term. Furthermore, the entries of $\mat{B}$ encode rich information regarding the effect that previous actions have on current ones. However, without additional considerations on the inherent structure of $\mat{B}$, fitting this model would require prohibitively many temporal observations. In Section~\ref{sec:autoregressive}, we demonstrate improvements over Hoff's approach by imposing low-rank tensor structure on the regression parameters.

\subsubsection{Image reconstruction}
\label{sec:motiv_pet}
Positron Emission Tomography (PET) is a major image modality widely used in hospitals as a tool for diagnosis and intervention \citep{ollingeretal97}. PET imaging provides detailed insights into metabolic processes within the body, making it invaluable for detecting and evaluating conditions such as cancer, neurological disorders, and cardiovascular diseases. Accurate reconstruction of these images is essential for making informed decisions regarding patient care, optimize treatment strategies, and evaluate interventions.

A PET scanner maps radioactive tracer concentrations from a subject into measurement data known as a sinogram. PET imaging attempts to reconstruct (i.e., estimate) an image matrix $\mat{B}$ corresponding to the subject that led to the observed sinogram $\mat{Y}$. More details on this process are provided in Section \ref{sec:pet}. A popular algorithm for PET reconstruction is the maximum likelihood-expectation maximization (ML-EM) algorithm of~\citet{sheppandvardi82} and its accelerated version~\citep{hudsonandlarkin94}. ML-EM is popular because it leverages the discrete nature of the data to enhance the image reconstruction's quality and reliability. 
ML-EM estimates a model that assumes that each entry of $\mat{Y}$ independently follows a Poisson distribution with rates corresponding to 
$\mathcal{R}(\mat{B})$--i.e., the discrete Radon transform of $\mat{B}$, which is of the same dimensions as $\mat{Y}$. Since the discrete Radon transform is a linear operator \citep{beylkin87}, the statistical model can be written as an instance of Equation \eqref{eq:poisreg}, with responses $\syi$ corresponding to entries in $\mat{Y}$ and covariates $\bxi$ corresponding to the Radon transform's basis.

Note that ML-EM progressively amplifies noise present across iterations if not properly regularized~\citep{snyderetal87,barrettetal94}. Several approaches have been proposed to remedy this problem, such as estimating the number of iterations \citep{veklerovandllacer87,selivanovetal01}, regularization \citep{depierroandyamagishi01,lange90}, and smoothing \citep{tomasiandmanduchi98,dabovetal07}.
In Section~\ref{sec:pet}, we show 
 that the inherent tensor structure in the PET reconstruction problem can be leveraged to significantly reduce the number of regression parameters and model reconstruction error.

\subsubsection{Change-point detection} 
\label{sec:motiv_changepoint}

Detecting significant change points in dyadic data allows us to identify substantial shifts in communication patterns between pairs of communicants over time. These change points can indicate critical moments when the nature of communication undergoes significant changes due to shifts in behavior or external influences. For example, the well-studied Enron Corpus~\citet{enron}---a collection of email messages between employees of the Enron Corporation leading up to the company's collapse in 2001---has been studied to identify if and when changes in communication patterns occurred between employees leading up to events that were later investigated for potential regulatory improprieties. In previous work, Bader et al. modeled the email communications by topic over time using low-rank tensor decompositions to characterize the temporal evolution of the data and illustrated several key changes over time~\citet{bader_enron}.

For the general change-point detection problem associated with communications data, consider a count tensor $\tYt$ of size $M_1\times M_2\times M_3$ encoding the number of times communications occur across $M_2$ senders, $M_3$ receivers, and $M_3$ topics at time $t$. We can model a change-point in communications at a known time $\tau$ as $\mbbE(\tYt) = \sxt\tB_1 + (1-\sxt) \tB_2$. Here  $t<\tau$ implies that $\tYt$ occurred before the change-point, and hence $\sxt=1$ and $\mbbE(\tYt) = \tB_1$. Similarly, $t>\tau$ implies that $\tYt$ occurred after the change-point, and hence $\sxt=0$ and $\mbbE(\tYt) = \tB_2$. These conditions are satisfied in Equation \eqref{eq:mpoisreg} if we use $\vecc(\tYt)$ as the response vector, and $[\sxt \ \ 1-\sxt]\transpose$ as the covariate vector. Although this framework assumes a known value of $\tau$, the model can be fit once for multiple values of $\tau$, and $\tau$ can be chosen as the one with the largest resulting loglikelihood. However, doing so would require estimating $\mat{B}$ in Equation \eqref{eq:mpoisreg} of size $2M_1M_2M_3$ multiple times, which can be challenging unless a very long temporal sequence of data is observed or we leverage the inherent tensor structure of $\tYt$. In Section \ref{sec:ptanova}, we develop a methodology that can leverage such tensor structure and demonstrate its use on synthetically generated communication data.

\subsection{Background and related work}
The previous examples highlight the limitations of training the model in Equation \eqref{eq:mpoisreg} on tensor-valued data without accounting for its inherent tensor structure. Indeed, in unconstrained vector-variate regression, the number of regression parameters in $\mat{B}$ increases with the dimensions of the vectorized tensor responses and covariates, leading to computational inefficiencies and potential overfitting.

Tensor decompositions enable the representation of multidimensional data in a more compact form by considering its inherent structure. The most common tensor decomposition methods include canonical polyadic (CP) decomposition \citep{carrollandchang70,harshman70}, Tucker decomposition \citep{tucker66}, and tensor train decomposition \citep{oseledets11}.  A tensor $\tT \in\mbbR^{M_1\times\dots\times M_P}$ has a CP decomposition of rank $R$ if 
\begin{equation}\label{eq:CPform1}
\tT = \sum_{r=1}^R \lambda_r \ba^{(1)}_r\circ\dots\circ \ba^{(P)}_r = [\![\blambda; \mat{A}^{(1)},\dots,\mat{A}^{(P)}]\!]\\
\end{equation}
where each $\ba^{(p)}_r$ denotes the $r$-th column of the  $M_p\times R$ matrix $\mat{A}^{(p)}$, $\circ$ denotes the vector outer product \citep{koldaandbader09}, and $[\![\blambda; \mat{A}^{(1)},\dots,\mat{A}^{(P)}]\!]$ with $\blambda=[\lambda_1,\dots,\lambda_R]$ is the compact notation of the decomposition. When $R$ is the smallest value such that Equation~\eqref{eq:CPform1} holds, the decomposition is \emph{exact}  (see~\citep{BaKo2025} for complete details). 
As the value of $R$ for exact decomposition of a tensor is in general unknown and computing it is NP-hard~\citep{HiLi13}, we can approximate a tensor with a \emph{low-rank CP decomposition} by minimizing $\|\tT-\tM\|_F$, subject to $\tM= [\![\blambda; \mat{A}^{(1)},\dots,\mat{A}^{(P)}]\!]$ for a value of $R$ that is much less than the sizes of the tensor dimensions $M_1,\dots, M_P$. Here $\|\cdot\|_F$ denotes the tensor Frobenius norm and is equivalent to the least squares loss function for fitting $\tM$ to $\tT$.

\subsubsection{The Poisson canonical polyadic (PCP) tensor model}\label{sec:pcp}
The Poisson canonical polyadic (PCP) tensor model \citep{chiandkolda12, lvdllp-arxiv:25} is a powerful statistical technique designed to estimate relationships in tensor count data via low-rank approximation of a tensor of Poisson parameters. Unlike traditional tensor decomposition methods that optimize the least squares loss between a tensor and a low-rank tensor decomposition, PCP assumes a Poisson noise model of the data and computes maximum likelihood estimates of the underlying Poisson parameters element-wise via approximate low-rank CP decomposition of the parameter tensor. 
The random tensor $\tY\in\mbbN^{M_1\times\dots\times M_P}$ follows a PCP distribution, written as
\begin{equation}\label{eq:PCP}
    \tY\sim\Poi(\tB),
\end{equation}
if each entry $\tYm$ independently follows a Poisson distribution $\Poi(\tBm)$, and $\tB \in \mbbRp^{M_1\times\dots\times M_P}$ has the form of a CP decomposition in Equation \eqref{eq:CPform1}. Here ${\boldsymbol m}$ denotes the element-wise tensor multi-index $(m_1,m_2,\dots,m_P)$ with $m_p \in \lbrace1,2,\dots,M_p\rbrace$ for $p=1,\dots,P$.

\subsubsection{Tensor-on-tensor regression (ToTR)}\label{sec:totr}
Unlike tensor decompositions, tensor regression models aim to understand and quantify relationships between tensor responses and/or covariates. Tensor regression reduces the parameter space and regularizes the regression coefficients by leveraging low-rank tensor decomposition techniques, ensuring stable parameter estimates. Several regression frameworks that efficiently allow for tensor responses or covariates (but not both) have recently been considered in \citep{zhouetal13,rabusseauandkadri16,sunandli17,liandzhang17,guhaniyogietal17,lietal18,zhouetal21}.

Tensor-on-tensor regression (ToTR) refers to the case where both the responses and covariates are tensors. Consider a tensor response  $\tYi \in \mbbR^{M_1\times \dots \times M_P}$ and a tensor covariate $\tXi \in \mbbR^{N_1\times \dots \times N_Q}$  (where $i=1,2,\dots,I$). An identity-link tensor-on-tensor regression model satisfies
\begin{equation}\label{eq:totr1}
\mbbE(\tYi) =  \XiB,
\end{equation}
where $\tB \in\mbbR^{N_1\times \dots \times N_Q\times M_1\times \dots \times M_P}$ is the regression coefficient tensor with the combined dimensions of the covariate $\tXi$ and response $\tYi$, and $\XiB \in \mbbR^{M_1\times \dots \times M_P}$ 
denotes the \emph{partial tensor contraction} of $\tXi$ onto $\tB$, defined element-wise as
\begin{equation}\label{eq:partialcontraction}
\XiBm = \sum_{n_1=1}^{N_1} \sum_{n_2=1}^{N_2} \cdots \sum_{n_Q=1}^{N_Q} \tXin \tBnm.
\end{equation}
where ${\boldsymbol m}$ and ${\boldsymbol n}$ are multi-indices and ${\boldsymbol n}, {\boldsymbol m}$ is a double multi-index over the combined response and covariate dimensions associated with the regression coefficients in $\tB$.

ToTR has been considered in \citep{hoff15,lock17,llosa18,raskuttietal19,liuetal20,gahrooeietal21,leeetal24,luoandzhang24}. These methodologies attempt to estimate the regression coefficients in Equation~\eqref{eq:totr1} by minimizing the least squares loss, making them appropriate for cases where each entry in $\tYi$ is independent, homoscedastic (constant variance), and Gaussian. \citet{llosaandmaitra22} considered the case where $\tYi$ follows a heteroscedastic tensor-variate normal distribution; \citep{akdemirandgupta11,hoff11}, and \citet{llosaandmaitra24} further considered the case where $\tYi$ has heavier or lighter tails than Gaussian. However, to our knowledge, existing ToTR models have primarily focused on continuous data and have not adequately accommodated discrete tensor responses, such as those of Section \ref{sec:motivation}.

\subsection{Overview and our contributions}
PCP decompositions (Section \ref{sec:pcp}) provide accurate and interpretable low-rank models of tensor count data by leveraging the statistical properties of the Poisson distribution. Meanwhile, ToTR (Section \ref{sec:totr}) models the complex relationships between tensor responses and covariates. In this article, we introduce \emph{Poisson-response tensor-on-tensor regression (PToTR)}, a novel method that combines the statistical properties of PCP decompositions with the supervised modeling capabilities of ToTR, extending the latter to handle discrete response data. To our knowledge, this is the first instance of ToTR being adapted for discrete data. PToTR emerges as a versatile tool applicable in various contexts, as demonstrated by the challenges addressed in Section \ref{sec:motivation}.

In section \ref{sec:ptotr} we formulate PToTR in detail, derive a maximum likelihood estimation algorithm for estimating the regression coefficients, and present a minimax lower bound on estimation error. These derivations will be used in Sections \ref{sec:autoregressive}, \ref{sec:pet}, and \ref{sec:ptanova}, where we will apply PToTR to the motivating problems of Section \ref{sec:motivation}. Finally, in Section \ref{sec:conclusion} we provide concluding remarks and ideas for future work.

\section{Poisson response tensor-on-tensor regression (PToTR)}\label{sec:ptotr}
Considering the ToTR framework of Equation \eqref{eq:totr1},  and following the PCP notation of \eqref{eq:PCP}, we formulate PToTR as
\begin{equation}\label{eq:ptotr}
    \tYi \overset{\text{indep.}}{\sim} \Poi(\XiB),
\end{equation}
where $\tB \in\mbbRp^{N_1\times \dots \times N_Q\times M_1\times \dots \times M_P}$ is the regression coefficient with the combined dimensions of the $i$-th covariate $\tXi \in \mbbRnn^{N_1\times \dots \times N_Q}$ and the $i$-th response $\tYi \in \mbbN^{M_1\times \dots \times M_P}$  ($i=1,2,\dots,I$). 
Upon vectorizing both sides of Equation \eqref{eq:ptotr}, we can write PToTR in the form of Equation \eqref{eq:mpoisreg}, where $\byi = \vecc(\tYi)$ is a vector response, $\bxi=\vecc(\tXi)$ is a vector covariate, and $\mat{B}$ is a large coefficient matrix  containing $\prod_pM_p\prod_qN_q$ entries of $\tB$ ($\mat{B}$ is a canonical matricization of $\tB$, see \citep[Table 1]{llosaandmaitra22}). 
While $\mat{B}$ has a large number of entries, it has inherent tensor structure that can be considered to restrict the parameter space. 
We will assume that $\tB$ has the form of a rank $R$ CP decomposition
\begin{equation}\label{eq:cpform2}
\tB = [\![\blambda; \mVq{1},\dots,\mVq{Q},\mUp{1},\dots,\mUp{P}]\!]
\end{equation}
with factor matrices $\mVq{q}\in\mbbRp^{N_q\times R}, q =1, \dots,Q$, and $\mUp{p}\in\mbbRp^{M_p\times R}, p =1,\dots,P$. 
This way, we reduce the number of parameters in $\tB$ from $\prod_pM_p\prod_qN_q$ to $R(\sum_pM_p+\sum_qN_q)$, making parameter inference possible in cases where a prohibitively large sample size $I$ would be required. 

\subsection{Identifiability and non-degeneracy}\label{sec:identif1}

Lack of identifiability in PToTR can occur  because any two factor matrices in $\{\mVq{1},\dots,\mVq{Q},\mUp{1},\dots,\mUp{P}\}$ can be scaled $\diag(\bgamma)$ and $\diag(\bgamma)^{-1}$ (where $\bgamma\in\mbbR^{R}$ is non-zero) and the PToTR model will remain unchanged. To ensure identifiability, we follow \cite{lvdllp-arxiv:25}  and scale the factor matrix columns to sum to one, so that for $\bone_R$ denoting a $R$-variate vector of ones
$$
\bone_R = \mVq{q}{\transpose}\bone_{N_q}=\mUp{p}{\transpose}\bone_{M_p}
$$
for all $q =1, \dots,Q$ and $p =1,\dots,P$.
Hence, the vector $\blambda$ absorbs all of the weights, and we order the entries such that
$\lambda_1 \geq \dots \geq \lambda_R$. 

A degenerate Poisson random variable in PToTR occurs if for any $\mi{m}$ and $i$, it holds that $\XiBm=0$. To avoid degeneracy, we will first assume that $\xi \coloneqq \min_i\|\tXi \|>0$, meaning that while $\tXi$ is strictly non-negative element-wise, it must contain at least one non-zero entry (such as it is the case for our illustrative examples of Section \ref{sec:motivation}). 
Second, we will assume that $\tB$ is strictly positive element-wise, which will be ensured by having the factors $\blambda,\mVq{1},\dots,\mVq{Q},\mUp{1},\dots,\mUp{P}$ be strictly positive element-wise. Next, we perform maximum likelihood over these factors and constraints.

\subsection{Maximum likelihood estimation}\label{sec:pToTRmle}
In this section, we present a maximum likelihood estimation algorithm designed to fit the PToTR model, where the loglikelihood function can be written as 
\begin{equation}\label{eq:llikpois}
 \ell(\tB) =  \sum_{i,\mi{m}}\!
 \Big[ \tYim  \log\left( \XiBm \right) - \XiBm \Big] - C \; .
\end{equation}
The constant $C=\sum_{i,\mi{m}}\log(\tYim !)$ does not depend on $\tB$ and hence it will be ignored in the remainder of this article.

\subsubsection{The optimization problem}\label{sec:MLEproblem}
To numerically approximate the maximum likelihood estimator (MLE) of $\tB$, we solve the optimization problem
\begin{equation}
\begin{split}
    &\max\  \ell(\tB) \quad \\
    &\text{ s.t.}\quad 
    \tB = [\![\blambda; \mVq{1},\dots,\mVq{Q},\mUp{1},\dots,\mUp{P}]\!] \in\Theta,
\end{split}
\end{equation}
where
$
\Theta = \Theta_\lambda \times \Omq{1} \times \dots \times \Omq{Q} \times \Psip{1} \times \dots \times \Psip{P}
$
with $\Theta_\lambda=(0,\infty)^R$, $\Omq{q}=\{\mVq{q}\in(0,1)^{N_q\times R}:\mVq{q}{\transpose}\bone_{N_q} = \bone_R\}$, and $\Psip{p}=\{\mUp{p}\in(0,1)^{M_p\times R}:\mUp{p}{\transpose}\bone_{M_p} = \bone_R\}$. As described in Section \ref{sec:identif1}, the constraint set  $\Theta$ is carefully chosen so that the fitted PToTR model remains identifiable and non-degenerate.

We extend the alternating optimization scheme introduced for PCP in~\cite{chiandkolda12}, which involves solving sub-problems that optimize one factor matrix while keeping all others fixed. This iterative approach can be viewed as an instance of a non-linear Gauss-Seidel algorithm \citep{bertsekasandtsitsiklis89} or a block-relaxation algorithm \citep{deleeuw94}. The sub-problems are addressed using majorization-minimization (MM) algorithms \citep{hunterandlange04} that we introduce next.

\subsubsection{The optimization sub-problems}\label{sec:MLEsubproblem}
Each sub-problem will be solved as an instance of the following theorem, which combines Theorem 2 of \citet{leeandseung00} and Theorem 4.3 of \citet{chiandkolda12} for application to PToTR. This algorithm can be used to numerically approximate the maximum likelihood estimator of $\mat{B}$ under the full-rank, vector-response, vector-covariate model of Equation \eqref{eq:mpoisreg}.
Note that $\oslash$, $*$, and $\log$ denote division, product, and logarithm applied element-wise to tensor operands.

\begin{theorem}\label{thm:nnstep}
	Let $\mat{Y}\in\mbbN^{J \times L}$, $\mat{C}\in\mbbRp^{J \times R}$, and $\mat{D}\in\mbbRp^{R \times L}$. Consider $
	f(\mat{C}) = \bone_{J}\transpose\left\{
	\mat{Y}*\log(\mat{C}\mat{D})
	-
	\mat{C}\mat{D}
	\right\}\bone_{L},
	$
	and $\mat{\Phi}(\mat{C}) = \left\{\left(
	\mat{Y}\oslash(\mat{C}\mat{D})
	\right)\mat{D}\transpose
	\right\}
	\oslash
	\left\{
	\bone_{J} \bone_{L}\transpose \mat{D}\transpose
	\right\}$.  Then the sequence $\{\mat{C}_{\{k\}}\}$ defined as 
	\begin{equation}\label{eq:MMupdate}
		\mat{C}_{\{k+1\}} \leftarrow
		\mat{C}_{\{k\}} * \mat{\Phi}(\mat{C}_{\{k\}}),
	\end{equation}
	is non-decreasing in $f$ when $\mat{C}_{\{0\}}>0$. Further, if $\mat{D}$ has linearly-independent rows, $f(\mat{C}_{\{0\}})<\infty$, and $\mat{\Phi}(\mat{C}_{\{0\}})>1$, then the sequence $\{\mat{C}_{\{k\}}\}$ will converge to a global maxima of $f(\mat{C})$.
\end{theorem}
 
The update in Equation \eqref{eq:MMupdate} is multiplicative. Hence, as long as the initial value $\mat{C}_{\{0\}}$ is contained in the constraint set  $(0,\infty)^{J\times R}$, the entire sequence $\{\mat{C}_{\{k\}}\}$ will remain in the constraint set . 
In Theorem \ref{thm:nnstep} denote $\by\transpose_j$ and $\bc\transpose_j$ the $j$-th row of $\mat{Y}$ and $\mat{C}$ respectively. If $\by\transpose_j\bone_{L}=0$  then $f(\mat{C})$ depends on $\bc\transpose_j$ only through $-\bc\transpose_j\mat{D}{\transpose}\bone_L$. Hence, a global maxima is attained at $\bc_j=0$ and it is not contained in the constraint set  $(0,\infty)^R$. In such cases we say that an MLE for $\bc_j$ does not exist (d.n.e). However, under $\mat{Y}\sim\Poi(\mat{C}\mat{D})$ element-wise,  
$$
P(\text{MLE for }\bc_j\text{ d.n.e.}) = \prod_{l=1}^LP(\mat{Y}_{j,l}=0) = e^{-\bc\transpose_j\mat{D}{\transpose}\bone_L}.
$$
Hence, the probability that a MLE for $\bc\transpose_j$ d.n.e decreases exponentially with the magnitudes and dimensions of $(\bc\transpose_j,\mat{D}$). In our regression setting, this probability will decrease exponentially with sample size $I$ and ambient dimensions.

\subsubsection{Estimation of response-sized factors \texorpdfstring{$\mUp{1},\dots, \mUp{P}$}{Lg}}\label{sec:mle_Us}

To estimate the pair $(\blambda,\mUp{p})$ for any $p=1,2,\dots,P$, first let $\mtUp{p} = \mUp{p}\diag(\blambda)$. Then we can define
$$
\widetilde{\mat{G}}_{ip} \coloneqq [\XiB]_{(p)} = \mtUp{p}\mat{G}_{ip},
$$
where $[\cdot]_{(p)}$ denotes the $p$-th mode tensor matricization \citep{BaKo2025},  $\mat{G}_{ip} = \diag(\bw_i)(\odot_{s\neq p} \mUp{s})\transpose$, $\bw_i =(\odot_{q} \mVq{q})\transpose\vecc(\tXi)$, and $\odot$ denotes the Khatri-Rao matrix product \citep{khatriandrao68}. 
This allows us to write the loglikelihood function of Equation \eqref{eq:llikpois}  as a function of $\mtUp{p}$ and fixed matrices $\mat{G}_{ip}$
\begin{equation}\label{eq:ptotr_U}
\begin{aligned}
\ell(\mtUp{p}) &=\!\!
&\sum_{i=1}^I \bone_{M_{p}}\transpose\left[
[\tYi]_{(p)}*\log(\widetilde{\mat{G}}_{ip}) - \widetilde{\mat{G}}_{ip}
\right]\bone_{M_{-p}} \; ,
\end{aligned}
\end{equation}
where $M_{-p} = M/M_p$ and $M=\prod_{s=1}^P M_s$. The above loglikelihood is of the form of $f(\mtUp{p})$ of Theorem \ref{thm:nnstep} after setting $\mat{Y} = \left[ [\tY_{(1)}]_{(p)},\dots,[\tY_{(I)}]_{(p)}] \right]$ and $\mat{D} = [\mat{G}_{1p}, \dots,  \mat{G}_{Ip}]$. Hence, we can use the multiplicative update of Equation \eqref{eq:MMupdate} which simplifies to the following non-decreasing updates 
\begin{equation}\label{eq:MLestU}
\begin{aligned}
\mtUp{p}_{\{k+1\}} \leftarrow
\mtUp{p}_{\{k\}}  &* 
\left\{\sum_{i=1}^I\left[\left(
[\tYi]_{(p)}\oslash(\mtUp{p}_{\{k\}}\mat{G}_{ip})
\right)\mat{G}_{ip}\transpose\right]
\right\}
\\&\oslash
\left\{
\bone_{M_p} (\sum_{i=1}^I \bw_i)\transpose
\right\} \; .  
\end{aligned}
\end{equation}
The above updates can be performed until the change in the loglikelihood, $|\ell(\mtUp{p}_{\{k+1\}}) - \ell(\mtUp{p}_{\{k\}})|$ is below some user-defined threshold.
The updates for $\blambda$ and $\mUp{p}$ come directly from $\mtUp{p}$. Stopping after $K$ iterations, we can then set 
$\widehat \blambda \leftarrow \bone_{M_p}\transpose\mtUp{p}_{\{K\}}$ and 
$\mhUp{p} \leftarrow \mtUp{p}_{\{K\}} \diag(\widehat\blambda)^{-1}$, ensuring that $\widehat\blambda\in\Theta_\lambda$ and $\mhUp{p}\in\mat{\Psi}_p$.

\subsubsection{Estimation of covariate-sized factors \texorpdfstring{$\mVq{1},\dots, \mVq{Q}$}{Lg}}\label{sec:mle_Vs}
To estimate the pair $(\blambda,\mVq{q})$ for any $q=1,2,\dots,Q$, first let $\mtVq{q} = \mVq{q}\diag(\blambda)$. Then we can define 
$$
\widetilde{\mat{H}}_{iq} \coloneqq \vecc(\XiB)= \mat{H}_{iq}\vecc(\mtVq{q}),
$$
where $\mat{H}_{iq}$ is a $\prod_pM_p\times RN_q$ matrix with the same elements rearranged from the $N_q\prod_pM_p\times R$ matrix 
$$
(\odot_{p} \mUp{p})\odot \mat{W}_{iq},\text{ where } \mat{W}_{iq} = [\tXi]_{(q)}(\odot_{s\neq q} \mVq{s}).
$$
This allows us to write Equation \eqref{eq:llikpois}  as a function of $\mtVq{q}$
\begin{equation}\label{eq:ptotr_V}
\ell(\mtVq{q}) =
\sum_{i=1}^I \bone_M\transpose\left[
	\vecc\left(\tYi\right)*\log(\widetilde{\mat{H}}_{iq}) - \widetilde{\mat{H}}_{iq}
	\right],
\end{equation}
where $M = \prod_{p=1}^P M_p$. The above loglikelihood is of the form of $f(\vecc(\mtVq{q})\transpose)$ of Theorem \ref{thm:nnstep} after setting 
$\mat{Y} = [(\vecc(\tY_{(1)})\transpose, \dots, (\vecc(\tY_{(I)})\transpose]$ and $\mat{D} = [\mat{H}_{1q}\transpose, \dots, \mat{H}_{Iq}\transpose]$. Hence, we can use the multiplicative update of Equation \eqref{eq:MMupdate}, which simplifies to the following non-decreasing updates when starting from $\mtVq{q}_{\{0\}} > 0$
\begin{equation}\label{eq:MLestV}
\begin{aligned}
\vecc(\mtVq{q}_{\{k+1\}}) \leftarrow &
\left\{\sum_{i=1}^I\left[\mat{H}_{iq}\transpose\left(
\vecc(\tYi)\oslash\left(\mat{H}_{iq}\vecc(\mtVq{q}_{\{k\}})\right)
\right)\right]
\right\} \\ &
*\vecc\left(\mtVq{q}_{\{k\}}\oslash
\sum_{i=1}^I\mat{W}_{iq} \right).
\end{aligned}
\end{equation}
As with $\mtUp{p}$ in the previous section, the above updates can be performed until the change in the loglikelihood is below some user-defined threshold. If we stop after $K$ iterations, $\blambda$ can be updated as 
$\widehat \blambda \leftarrow \bone_{N_q}\transpose\mtVq{q}_{\{K\}}$ 
and $\mVq{q}$ can be updated as $\mhVq{q} \leftarrow \mtVq{q}_{\{K\}} \diag(\widehat\blambda)^{-1}$. This ensures that $\widehat\blambda\in\Theta_\lambda$ and $\mhVq{q}\in\Omq{q}$.

\subsubsection{Estimation algorithm}

The complete algorithm for maximum likelihood estimation of $\tB$ in a PToTR model is summarized in Algorithm \ref{alg:PToTR}.

\begin{algorithm}[H]
\scriptsize
\begin{algorithmic}[1]
\REQUIRE $\tYi \in \mbbN^{M_1\times \dots \times M_P}$ and $\tXi \in \mbbRnn^{N_1\times \dots \times N_Q}$  $(i=1,2,\dots,I)$);\\
$\mUp{p} \in \mbbRp^{M_p \times R} \; (p=1,\dots,P)$; $\mVq{q}  \in \mbbRp^{N_q \times R} \; (q=1,\dots,Q)$ 
\WHILE{convergence is not met}
\FOR{$p=1$ to $P$}
\STATE $\mtUp{p} \leftarrow \mUp{p} \text{diag}(\hat\blambda)$
\WHILE{convergence is not met}
\STATE    update $\mtUp{p}$ according to Equation \eqref{eq:MLestU}
\ENDWHILE
\STATE $\widehat \blambda \leftarrow \bone_{M_p}\transpose\mtUp{p}, \; \mhUp{p} \leftarrow \mtUp{p} \diag(\widehat\blambda)^{-1}$
\ENDFOR
\FOR{$q=1$ to $Q$}
\STATE $\mtVq{q} \leftarrow \mhVq{q} \text{diag}(\widehat\blambda)$
\WHILE{convergence is not met}
\STATE    update $\mat{\tilde V}_q$ according to Equation \eqref{eq:MLestV}
\ENDWHILE
\STATE $\widehat \blambda \leftarrow \bone_{N_q}\transpose\mtVq{q}, \;  \mhVq{q} \leftarrow \mtVq{q} \diag(\widehat\blambda)^{-1}$
\ENDFOR
\STATE $\thB = [\![\widehat\blambda; \mhVq{1},\dots,\mhVq{Q},\mhUp{1},\dots,\mhUp{P}]\!]$
\IF{convergence of $\thB$ is met} \label{alg:ptotr:Bconv}
\STATE \textbf{break}
\ENDIF
\ENDWHILE
\end{algorithmic}
\caption{Maximum likelihood estimation for PToTR}
\label{alg:PToTR}
\end{algorithm}
We initialize the algorithm with uniformly sampled entries in the constraint set  $\Theta$. This ensures that the regularity conditions of \citep[Theorem 1]{deleeuw94} hold, and hence that our algorithm converges.
We assess convergence of $\thB$ in Line~\ref{alg:ptotr:Bconv} using a relative change in the loglikelihood specified in Equation \eqref{eq:llikpois}. 
We can check that an MLE exists associated with each $\mUp{p}$ ($p=1,\dots,P$) by checking that $(\sum_i[\tYi]_{(p)})\bone_{M_{-p}}$ does not contain zero entries. As described in Section \ref{sec:MLEsubproblem}, the probability of this happening decreases exponentially with the magnitude of $\blambda$, and the sample size $I$. MLEs for $\mVq{q}$ ($q=1,\dots,Q$) will exist as long as a non-zero entry is observed in any observation $\tYi$ ($i=1,\dots,I$).

\subsection{Minimax lower bound}
In this section, we establish the minimax lower bound on the PToTR estimator error. We also provide remarks on the implications and and a proof of the bound.

\begin{theorem}{(Minimax lower bound)}\label{thm:minimax}
	Under the PToTR model of Equation \eqref{eq:ptotr} for responses with equal dimension sizes (${\bar M}\!\coloneqq\! M_1 \!=\! \dots\!=\!M_P$), covariates with equal dimension sizes (${\bar N}\!\coloneqq\! N_1 \!=\! \dots\!=\!N_Q$), and non-negative $\mat{X}=[\vecc(\tX_{(1)})\dots \vecc(\tX_{(I)}))]$, denote $\text{inf}_{\thB}$ as the infimum over all estimators $\thB \in  S_R(\beta,\alpha) $, where
	\begin{equation}\label{eq:minimax_S_R}
		\begin{aligned}
               S_R(\beta,\alpha) &\coloneqq 
               \Big\{ \tT = [\![\blambda; \mat{D}^{(1)},\dots,\mat{D}^{(Q)},\mat{C}^{(1)},\dots,\mat{C}^{(P)}]\!] 
               \\&:\; \mat{D}^{(q)}\in\mbbRp^{\bar{N} \times R}, 
               \mat{C}^{(p)}\in\mbbRp^{\bar{M} \times R},
               \beta<\tT_{\mi{n},\mi{m}}\leq \alpha\Big\}.
		\end{aligned}
	\end{equation}
	Assume that $R\leq \min(\bar{N},\bar{M})$, $J\coloneqq \max(\bar{N},\bar{M})>16$, 
	$\xi \coloneqq \min_i||\tXi||_1>0$, and
	\begin{equation}\label{eq:epsilon_condition}
		\dfrac{\beta\log 2}{(\alpha - \beta)^2} 
		\Big(\dfrac{JR}{16}-1\Big) 
		\dfrac{\xi}{\|\mat{X}\|_2^2}\leq \bar{N}^Q\bar{M}^P,
	\end{equation}
    where $\|\mat{X}\|_2$ denotes the spectral norm of $\mat{X}$. 	Then, 
	$$
	\inf_{\hat \tB}\sup_{\tB\in S_R(\beta,\alpha)}\mbbE(|| \thB-\tB||^2_F)\geq
	\dfrac{\beta\log 2}{128}
	\Big(\dfrac{JR}{16}-1\Big)\dfrac{\xi}{\|\mat{X}\|_2^2}.
	$$
\end{theorem}

\subsubsection{Remarks}

The term
$
\frac{\beta\log2}{128}\Bigl(\frac{JR}{16}-1\Bigr)
$
is a consequence of the particular choice of packing set over $S_R(\beta,\alpha)$ used in the proof. This term shows that no estimator can achieve worst‐case squared error smaller than a constant multiple of $JR\beta$, meaning that the minimax risk in estimating $\tB$ grows with the low-rank factor dimension $JR$ and not the tensor dimension $\bar{N}^Q\bar{M}^P$.
The multiplicative factor
$\xi/\|\mat{X}\|_2^2$
arises from bounding the KL divergence between probability distributions indexed by elements of the packing set.  Here 
$\xi$ ensures all Poisson rates are bounded away from zero,
while $\|\mat{X}\|_2^2$ measures the maximum amount of PToTR noise that can be amplified when computing the Poisson rates in $\langle \tX_{(i)}|\tB\rangle$. Indeed, because $\|\mat{X}\|_2^2\le\|\mat{X}\|_F^2$, the worst-case squared error decreases proportionally with sample size $I$ and covariate dimensions ${\bar N}^Q$.   

Furthermore, Theorem~\ref{thm:minimax} states that to achieve a target error 
$\mathbb{E}(\|\widehat \tB-\tB\|_F^2)\le\varepsilon$,
 one must necessarily have
$$
\|\mat{X}\|_2^2\geq
\frac{\beta\log2}{128}\Bigl(\frac{JR}{16}-1\Bigr)\,
\frac{\xi}{\varepsilon},
$$
so that no estimator—MLE or otherwise—can evade this sample‐complexity requirement.

Finally, we note that under the special case $I=1$ and $\tX_{(1)}=1$, our PToTR model reduces to the PCP model of \citep{lvdllp-arxiv:25,lopezetal25}. In this case, we have $\xi/\|\mat{X}\|_2^2=1$ and our lower bound reduces to the minimax bound of \cite[Theorem 5]{lopezetal25}. 

In summary, the minimax bound on the PToTR estimator error in Theorem~\ref{thm:minimax} demonstrates that the fundamental difficulty of estimating $\tB$ in PToTR is governed by the factor dimension $JR$ and the spectral norm of the matrix of vectorized covariates $\mat{X}$.

\subsubsection{Proof of Theorem \ref{thm:minimax}}

We first state the Generalized Fano Method \citep[Prop.15.12]{wainwright19}---a standard approach for establishing minimax lower bounds on estimator error---adapted for PToTR. This adaptation is an extension of Theorem 9 and Corollary 3 from \citep{lopezetal25}.
Throughout this proof we will assume that  $\by \coloneqq [\vecc (\tY_{(1)})^{\transpose}
\dots  \vecc (\tY_{(I)})^{\transpose}
]^{\transpose}$.

\begin{theorem}[Generalized Fano Method for PToTR]\label{thm:fano}
Given a finite $\mathcal{F} \subset \mathcal{S}_R(\beta, \alpha)$, as defined in Equation~\eqref{eq:minimax_S_R}, let $G=|\mathcal{F}|$ and denote $p(\by \!\mid\! \tB)$ the PMF of $\by$ under the PToTR model of Equation \eqref{eq:ptotr}. Assume 
\begin{enumerate}
\item (Separation) there exists $\delta > 0$ such that
$$
\min_{k \ne k',\tB^{(k)},\tB^{(k')}\in\mathcal{F}}
\|\tB^{(k)} - \tB^{(k')}\|_F
\ge
\delta \; ; \mbox{and}
$$

\item (KL Control) there exists $\gamma > 0$ such that
\[
\max_{k \ne k',\tB^{(k)},\tB^{(k')}\in\mathcal{F}}
D_{KL}\!\left(
p(\by \mid \tB^{(k)})
\;\|\;
p(\by \mid \tB^{(k')})
\right)
\le
\gamma.
\]
\end{enumerate}

Then for any estimator $\widehat{\tB}\in S_R(\beta, \alpha)$, we have
\begin{equation}\label{eq:fanoineq}
\inf_{\widehat{\tB}}
\sup_{\tB \in\mathcal{S}_R(\beta, \alpha)}
\mbbE
(\|\widehat{\tB} - \tB\|_F^2)
\ge
\frac{\delta^2}{4}
\left(
1 -
\frac{\gamma + \log 2}{\log G}
\right).
\end{equation}
\end{theorem}

Theorem~\ref{thm:fano} is a direct translation of the method as described in Proposition 15.12 and Equation 15.34 of~\citet{wainwright19}, with $\mathcal{F}$ representing the $\delta$-separated set, $\Phi(\delta) = \delta^2$, and each  $\tB^{(j)}$ associated with a PToTR probability distribution across all the entries of $\by$. 
Thus, the proof of Theorem~\ref{thm:fano} follows from \citep[\S15.4]{wainwright19} and is not reproduced here in the interest of space.

To apply the Generalized Fano Method to PToTR, we first define a packing set for $S_R(\beta,\alpha)$ from Equation~\eqref{eq:minimax_S_R} in Lemmas~\ref{lemma:vgb}~and~\ref{lemma:pack} below. Next, we establish a bound on the KL divergence between elements of that packing set in Lemma~\ref{lemma:boundKL}.

\begin{lemma}[Varshamov-Gilbert Bound {\citep[Lemma~7]{wangandli20}}]\label{lemma:vgb}
    Let $\Omega = \left\{(w_1, w_2, \ldots, w_m) \ | \ w_i \in \{0, 1\}\right\} \subseteq \mbbR^m$, and take $m > 8$. 
    Then, there exists a subset $\left\{\bw^{(0)}, \bw^{(1)}, \ldots, \bw^{(L)}\right\}\subset\Omega$ such that $\bw^{(0)}$ is the zero vector and for $0 \leq k < k' \leq L$,
    $$\left\| \bw^{(k)} - \bw^{(k')} \right\|_0 \geq \frac{m}{8},$$
    where $\| \cdot \|_0$ denotes Hamming distance and $L \geq 2^{m / 8}$. 
\end{lemma}

\begin{lemma}[Minimax Packing Set]\label{lemma:pack}
Suppose $R \le \min(\bar{N},\bar{M})$ and let $J := \max(\bar{N},\bar{M})$.
For a constant $\varepsilon \in (0,1]$ and bounds $0 < \beta \le \alpha$,  there exists a finite set $\mathcal{F} \subseteq \mathcal{S}_R(\beta,\alpha)$ such that $|\mathcal{F}| \ge 2^{J R}$ , and for any distinct $\tB^{(k)}, \tB^{(k')} \in \mathcal{F}$,
\begin{equation}\label{eq:Upack}
\|\tB^{(k)} - \tB^{(k')}\|_F
\ge
\frac{\varepsilon}{4}
(\alpha - \beta)
\sqrt{\bar{N}^Q \bar{M}^P},
\end{equation}
and
\begin{equation}\label{eq:Lpack}
\|\tB^{(k)} - \tB^{(k')}\|_F
\le
\varepsilon
(\alpha - \beta)
\sqrt{\bar{N}^Q \bar{M}^P}.
\end{equation}
\end{lemma}

\begin{proof}
Define 
$
\mathcal{C}
=
\left\{
\mat{C}\in \mbbRp^{J \times R}
:
\mat{C}_{j,r} \in \{\beta,\,
\beta + \varepsilon(\alpha-\beta)\}
\right\}.
$
Any matrix $\mat{C}\in\mathcal{C}$ can be identified by a binary $JR-$dimensional vector, and hence by the Varshamov–Gilbert bound of Lemma \ref{lemma:vgb}, 
there exists $\widetilde{\mathcal{C}} \subset \mathcal{C}$
such that $
|\widetilde{\mathcal{C}}|
\ge
2^{J R}
$
and for distinct $\mat{C}^{(k)}, \mat{C}^{(k')} \in \widetilde{\mathcal{C}}$,
\begin{equation*}
\|\mat{C}^{(k)} - \mat{C}^{(k')}\|_0
\ge
\frac{J R}{8}.
\end{equation*}
Since each differing entry contributes
$\varepsilon(\alpha-\beta)$
to the Frobenius norm, we have
\begin{equation*}
\|\mat{C}^{(k)} - \mat{C}^{(k')}\|_F^2
\ge
\varepsilon^2(\alpha-\beta)^2
\frac{J R}{8}.
\end{equation*}
Now, let $\mat{\Pi}(\mat{C})$ denote the matrix obtained by repeating the matrix $\mat{C}$
horizontally a fixed number of times (and padding with the first column of $\mat{C}$ if necessary)
so that $\mat{\Pi}(\mat{C}) \in \mbbRp^{J \times \min(\bar{N},\bar{M})}$. Because $\mat{\Pi}(\mat{C})$ consists of at most $\min(\bar{N},\bar{M})/R$ repeated blocks of $\mat{C}$, we have 
\begin{equation}\label{eq:proof_AofC}
\begin{aligned}
\|\mat{\Pi}(\mat{C}^{(k)}) - \mat{\Pi}(\mat{C}^{(k')})\|_F^2
&\ge
\Big\lfloor \frac{\min(\bar{N},\bar{M})}{R}\Big\rfloor
\|\mat{C}^{(k)} - \mat{C}^{(k')}\|_F^2 \\
&\ge
\frac{\min(\bar{N},\bar{M})}{2R}
\varepsilon^2(\alpha-\beta)^2
\frac{J R}{8}
\\&=
\dfrac{\bar{N}\bar{M}}{16}\varepsilon^2(\alpha-\beta)^2
\end{aligned}
\end{equation}
where we used $\lfloor x\rfloor/x\geq 1/2$ for $x\geq 1$ and $J\min(\bar{N},\bar{M}) = \bar{N}\bar{M}$.
Now we can introduce $\mathcal{F}$. For $\mat{\Pi}^*(\mat{C}) = \mat{\Pi}(\mat{C})$ if $\bar{M}<\bar{N}$ and $\mat{\Pi}^*(\mat{C}) = (\mat{\Pi}(\mat{C}))^\top$ otherwise, define 
 $$
\mathcal{F}
=
\{\bone_{\bar N} \underbrace{\circ\dots\circ}_{Q-1 \text{ times}} \bone_{\bar N} 
\circ  \mat{\Pi}^*(\mat{C}) \circ 
\bone_{\bar M} \underbrace{\circ\dots\circ}_{P-1 \text{ times}} \bone_{\bar M} 
\;:\;
\mat{C}\in \widetilde{\mathcal{C}}
\}.
$$
Clearly, $
|\mathcal{F}|
=
|\widetilde{\mathcal{C}}|
\ge
2^{J R}.
$
Since any $\mat{C}\in\widetilde{\mathcal{C}}$ is of rank at most $R$ and $\mat{\Pi}(\mat{C})$ is formed by block repetition of $\mat{C}$, $\mat{\Pi}(\mat{C})$ is also of rank at most $R$, and therefore tensors in $\mathcal{F}$ are also of rank at most $R$. Furthermore, since tensors in $\mathcal{F}$ have entries $\beta,\beta + \varepsilon(\alpha-\beta)\in[\beta,\alpha]$, we have that 
$\mathcal{F}\subseteq \mathcal{S}_R(\beta,\alpha)$. To establish the lower bound, for distinct $\tB^{(k)}, \tB^{(k')} \in \mathcal{F}$ with corresponding $\mat{C}^{(k)},\mat{C}^{(k')}\in \widetilde{\mathcal{C}}$, from Equation \eqref{eq:proof_AofC} we have  
\begin{equation*}
\begin{aligned}
\|\tB^{(k)} - \tB^{(k')}\|_F^2
&=
\bar{N}^{Q-1}\bar{M}^{P-1}\|\mat{\Pi}(\mat{C}^{(k)}) - \mat{\Pi}(\mat{C}^{(k')})\|_F^2
\\&\ge 
\dfrac{\bar{N}^Q\bar{M}^P}{16}\varepsilon^2(\alpha-\beta)^2.
\end{aligned}
\end{equation*}
The lower bound follows directly since
entries differ by at most
$\varepsilon(\alpha-\beta)$
over at most $\bar{N}^Q \bar{M}^P$ positions.
\end{proof}

\begin{lemma}[KL Divergence Bound]\label{lemma:boundKL}
Consider the model $\byi \sim\text{Poisson}(\mat{B} \bxi)$ of Equation \eqref{eq:mpoisreg}, where all the entries of $\by = [\by_{(1)},\dots,\by_{(I)}]^\top$ are independent of each other, $\bxi\in\mbbRnn^L$, and call $\mat{X} = [\bx_{(1)},\dots,\bx_{(I)}]$. Suppose that $\xi := \min_i \|\bxi\|_1 > 0$
and that for some $\beta>0$,
$$
\mat{B}\in S(\beta)\coloneqq \{ \mat{C} \in\mbbRp^{J\times L}| \mat{C}_{j,l} \geq \beta\}.
$$
For $\mat{B}^{(k)} \!=\! [\bb_{1}^{(k)},\dots,\bb_{J}^{(k)}]^\top,\mat{B}^{(k')} \!=\! [\bb_{1}^{(k')},\dots,\bb_{J}^{(k')}]^\top\in S(\beta)$ that are distinct, we have the following bound on the KL divergence
\begin{equation*}
D_{KL}\!\left( p(\by\mid \mat{B}^{(k)})\,\|\, p(\by\mid \mat{B}^{(k')}) \right)
\le
\frac{\|\mat{X}\|_2^2}{\beta \xi}
\|\mat{B}^{(k)} - \mat{B}^{(k')}\|_F^2 .
\end{equation*}
\end{lemma}

\begin{proof}
Denote the $d$-th entry of the vector $\byi$ as $y_{(i)d}$. Using independence in $\by$ and $\mbbE_{\mat{B}^{(k)}}(y_{(i)d}) = {\bxi^\top \bb_{d}^{(k)}}$, we have
\begin{equation*}
\begin{aligned}
D_{KL}&(p(\by\mid \mat{B}^{(k)})\|p(\by\mid \mat{B}^{(k')}))
\coloneqq 
\mbbE_{\mat{B}^{(k)}}\Big(\log\dfrac{p(\by\mid \mat{B}^{(k)})}{p(\by\mid \mat{B}^{(k')})}\Big)
\\&=
\sum_{d,i}
\left[
{\bxi^\top \bb_{d}^{(k)}}
\log\!\left(\frac{{\bxi^\top \bb_{d}^{(k)}}}{{\bxi^\top \bb_{d}^{(k')}}}\right)
- {\bxi^\top \bb_{d}^{(k)}}
+ {\bxi^\top \bb_{d}^{(k')}}
\right] \\
&\le
\sum_{d,i}
\frac{({\bxi^\top \bb_{d}^{(k)}} - {\bxi^\top \bb_{d}^{(k')}})^2}{{\bxi^\top \bb_{d}^{(k')}}}\\
&\le
\frac{1}{\beta \xi}
\sum_{d,i}
({\bx_i^\top \bb_{d}^{(k)}} - {\bx_i^\top \bb_{d}^{(k')}})^2 \\
&=
\frac{1}{\beta \xi}
\sum_d
(\bb_{d}^{(k)}-\bb_{d}^{(k')})^\top
\mat{X} \mat{X}^\top
(\bb_{d}^{(k)}-\bb_{d}^{(k')}) \\
&\le
\frac{\|\mat{X}\|_2^2}{\beta \xi}
\sum_d \|\bb_{d}^{(k)}-\bb_{d}^{(k')}\|_2^2\\
&=
\frac{\|\mat{X}\|_2^2}{\beta \xi}
\|\mat{B}^{(k)} - \mat{B}^{(k')}\|_F^2,
\end{aligned}
\end{equation*}
where the first inequality used $\log(x) \leq x-1$ and rearranged terms, the second inequality used 
$$
{\bxi^\top \bb_{d}^{(k)}}\geq \beta \bone_L^\top \bxi = \beta||\bxi||_1\geq \beta\xi,
$$ and the third inequality used that $||\mat{X}\transpose\bc||_2
\!\leq\! ||\mat{X}||_2||\bc||_2$ for any $\bc$, which results from the definition of spectral norm.
\end{proof}

With these lemmas now in hand, we are now ready to prove Theorem~\ref{thm:minimax}.
First, Lemma \ref{lemma:pack} states that $\mathcal{F} \subseteq \mathcal{S}_R(\beta,\alpha)$ and it is finite with $G\coloneq |\mathcal{F}|\ge 2^{J R}$. Furthermore, according to Equation \eqref{eq:Upack} the separation condition of Theorem \ref{thm:fano} is satisfied with $\delta = \frac{\varepsilon}{4}
(\alpha - \beta)
\sqrt{\bar{N}^Q \bar{M}^P}$. Because our PToTR model can be written in the vector-response, vector-covariate form of Equation \eqref{eq:mpoisreg}, we can invoke Lemma \ref{lemma:boundKL} to obtain an upper bound on the KL divergence. That is, for any distinct $\tB^{(k)}, \tB^{(k')} \in \mathcal{F}$, we have 
\begin{equation*}
    \begin{aligned}
    D_{KL}\!\big(
p(\by \mid \tB^{(k)})
\;\|\;
p(\by \mid &\tB^{(k')})
\big)
\le
\frac{\|\mat{X}\|_2^2}{\beta \xi}
\|\tB^{(k)} - \tB^{(k')}_2\|_F^2
\\& \le
\underbrace{\frac{\|\mat{X}\|_2^2}{\beta \xi}
\varepsilon^2
(\alpha - \beta)^2
\bar{N}^Q \bar{M}^P}_{\gamma},
    \end{aligned}
\end{equation*}
where the first inequality is a result of Lemma \ref{lemma:boundKL} and the second inequality is from Equation \eqref{eq:Lpack}. Therefore, Equation \eqref{eq:fanoineq} holds for these values of $\delta,\gamma$, and any value of $\varepsilon \in (0,1)$. 

Let
$$
\varepsilon^2 = 
\dfrac{\beta\log 2}{\bar{N}^Q\bar{M}^P(\alpha - \beta)^2} 
\Big(\dfrac{JR}{16}-1\Big) 
\dfrac{\xi}{\|\mat{X}\|_2^2},
$$
which satisfies $\varepsilon \in (0,1)$ due to $JR>16$ and Equation \eqref{eq:epsilon_condition}. Therefore, we have that 
$\gamma = (\frac{JR}{16}-1) \log 2$
and hence 
$$
\frac{\gamma + \log 2}{\log G}
\leq 
\frac{\gamma + \log 2}{JR\log 2}
\leq \dfrac{1}{2}.
$$
Substituting $\varepsilon^2$ into Equation  \eqref{eq:fanoineq} leads to
\begin{equation*}
\begin{aligned}
\inf_{\widehat{\tB}}
\sup_{\tB \in \mathcal{F}}
\mbbE
(\|\widehat{\tB} - \tB\|_F^2)
&\ge
\frac{\delta^2}{4}
\left(
1 -
\frac{\gamma + \log 2}{\log G}
\right).
\\&\ge
\frac{\delta^2}{8}
=
\frac{\varepsilon^2}{128}
(\alpha - \beta)^2
\bar{N}^Q \bar{M}^P
\\&=\dfrac{\beta\log 2}{128}
	\Big(\dfrac{JR}{16}-1\Big)\dfrac{\xi}{\|\mat{X}\|_2^2},
\end{aligned}
\end{equation*}
which completes the proof. \qed

With an estimation algorithm and minimax bounds for PToTR, we are now ready to revisit the motivating application problems of Section \ref{sec:motivation}.

\section{Applications}

\subsection{Longitudinal relational data analysis}\label{sec:autoregressive}
Section \ref{sec:motiv_longitudinal} laid out a tensor autoregressive model for predicting future actions in the ICEWS database. In general for longitudinal relational data analysis, we encode relations across $M_3$ kinds of actions that $M_1$ objects have towards (possibly the same) $M_2$ objects as a time series of tensors $\{\tYt\}$, where each $\tYt$ is a tensor of size $M_1 \times M_2 \times M_3$. The entries of the tensor $\tYt$ represent the directed actions that occur across all the pairs of objects (dyads). For example, in Section \ref{sec:motiv_longitudinal}, the entry $\tYtm$, where $\mi{m}=(m_1,m_2, m_3)$, is a count of the number of times the action $m_3$ was taken by country $m_1$ towards country $m_2$ during week $t$.

 Hoff introduced an instance of ToTR in order to analyze such longitudinal relational data~\citet{hoff15}. His analysis assumed Gaussian responses, and hence applied quantile-to-quantile transformations on the count data before fitting a ToTR model. In this section, we demonstrate the use of PToTR on such data, which allows us to 1) model the data as Poisson random variables without the need to perform lossy transformations, and 2) utilize a more flexible model of the regression coefficients.

\subsubsection{Methodology} 

One approach towards the analysis of longitudinal data is the autoregressive model, which models a time series at time $t$ as a function of data at previous time $t-1$.  In the context of PToTR we can write a Poisson-response autoregressive model as
\begin{equation}\label{eq:autoreg1}
    \tYt \sim \Poi (\langle \tXt | \tB \rangle) ,\quad t = 1,2,\dots,T,
\end{equation}
where $\tXt = \tYtmo$ but will be modified later to incorporate other longitudinal effects. Here $\tB$ contains information regarding how an action at time $t$ is affected by the actions that occurred at time $t-1$. For example, for a given pair of countries $(m_1,m_2)$ and action $m_3$ we have the following conditional expectation
\begin{equation}\label{eq:autoreg2}
    \mbbE(\tYtm | \tYtmo) = \sum_{\mi{n}} \tBnm\tYtmon.
\end{equation}
Hence, $\tBnm$ describes the effect that $\tYtmon$ has on $\tYtm$, or how action $n_3$ by country $n_1$ towards $n_2$ at one time affects action $m_3$ by country $m_1$ towards $m_2$ at the next time. In other words, the tensor $\tB$ is very rich in information, as it describes all the factors and their interactions at the cost of $M_1M_2M_3N_1N_2N_3=M_1^2M_2^2M_3^2$ parameters, which makes estimation intractable unless the number of temporal observations $T$ is extremely large. 

\paragraph{Parameter reduction and previous work}

The simplest multiplicative model that describes $\tB$ with only $2(M_1+M_2+M_3)$ parameters assumes that Equation \eqref{eq:autoreg2} holds with
\begin{equation}\label{eq:simplemult}
    \tBnm = \bv^{(1)}_{n_1} \bv^{(2)}_{n_2} \bv^{(3)}_{n_3} 
    \bu^{(1)}_{m_1}  \bu^{(2)}_{m_2}  \bu^{(3)}_{m_3} ,
\end{equation}
where the vectors $\bu^{(p)},\bv^{(p)}\in\mbbRp^{M_p}$ ($p=1,2,3$). This multiplicative model is generally too simple because, for example, it does not consider the effect that actions taken by country $n_1$ (contained in $\bv^{(1)}_{n_1}$) have on the previous actions taken by country $m_1$ (contained in $\bu^{(1)}_{m_1}$).
Such interactions were considered  in~\citet{hoff15} which assumed that Equation \eqref{eq:autoreg2} holds with $\tBnm = \mat{U}^{(1)}_{n_1,m_1} \mat{U}^{(2)}_{n_2,m_2} \mat{U}^{(3)}_{n_3,m_3}$, where each $\mat{U}^{(p)}$ is a square matrix of size $M_p \times M_p$ ($p=1,2,3$). This means that the effect captured by $\tBnm$ is assumed to be the multiplicative effect of three components: $ \mat{U}^{(1)}_{n_1,m_1}$ that describes how the actions of $m_1$ are influenced by the previous actions of $n_1$, $\mat{U}^{(2)}_{n_2,m_2}$ that describes how the actions towards $m_2$ are influenced by the previous actions towards $n_2$, and $\mat{U}^{(3)}_{n_3,m_3}$ that describes how the actions towards $m_3$ are influenced by the previous actions towards $n_3$. This multiplicative model generalizes that of Equation
\eqref{eq:simplemult}, which becomes the special case where $\mat{U}^{(1)},\mat{U}^{(2)},\mat{U}^{(3)}$ are rank one matrices. This multiplicative model reduces the overall number of parameters involved in $\tB$ from $M_1^2M_2^2M_3^2$ to $M_1^2+M_2^2+M_3^2$, which makes estimation possible with many fewer observations. However, this multiplicative model is also restrictive because it does not take into account interactions that might occur between countries or actions.

In the context of ToTR, \citet{llosaandmaitra22} demonstrated that the multiplicative model of \citet{hoff15} is an instance of ToTR using an outer-product (OP) factorization model of $\tB$ . Instead of the OP model, here we assume that $\tB$ has the rank $R$ CP decomposition model of Equation \eqref{eq:cpform2}, which reduces the number of parameters from $M_1^2M_2^2M_3^2$ to $2R(M_1+M_2+M_3)$. Hence, we assume that $\tBnm = \sum_{r=1}^R \mVq{1}_{n_1,r} \mVq{2}_{n_2,r} \mVq{3}_{n_3,r} \mUp{1}_{m_1,r} \mUp{2}_{m_2,r} \mUp{3}_{m_3,r}$. The matrices $\mVq{p}$ and $\mUp{p}$ are of size $M_p \times R$ ($p=1,2,3$) and are not easily interpretable like in the previous models, but together the matrices describe the complex interactions that exist in the tensor $\tB$ though the combination of additive and multiplicative effects. The simple multiplicative effect model of Equation \eqref{eq:simplemult} is the special case where $R=1$, and the general case with rank $R$ assumes that $\tB$ is the addition of $R$ of these multiplicative terms. Hence, the CP model of $\tB$ allows us to study the relationships in the data as a whole with all their nested interactions and  allows us to  adjust the level of desired recovery by adjusting the rank $R$.

\paragraph{Integrating trend and long-term dependence}
The autoregressive model as described in Equation \eqref{eq:autoreg1} (where $\tXt = \tYtmo$) is not stationary. 
To see this, suppose first-order stationarity holds---i.e., $\tM \coloneqq \mbbE(\tYt) = \mbbE(\tYtmo)$. Then,
\begin{equation*}
\begin{aligned}
\tM &\coloneqq \mbbE(\tYt) 
\\  &= 
\mbbE(\mbbE(\tYt|\tYtmo))
\\  &= 
\mbbE(\langle\tYtmo|\tB\rangle)
=
\langle\tM|\tB\rangle,
\end{aligned}
\end{equation*}
which implies that either $\tM=0$ or $\tB$ is an identity operator. A similar calculation shows that if $\tB$ is an identity operator, then $\text{Var}(\tYtm) = \text{Var}(\tY_{(t-1)\boldsymbol{m}}) + \tM_{\boldsymbol{m}}$, which means that second-order stationarity (i.e., constant variance) doesn't hold unless $\tM =0$, and this is not possible for Poisson rates.
In the Gaussian case this is typically resolved by first removing the trend, so that indeed $\tM=0$. This is not possible under the Poisson-response assumptions, as it would lead to continuous responses. For these reasons, we will incorporate trend as part of the model by modifying the covariate $\tXt$.

An $F$-degree polynomial trend can be integrated in the PToTR autoregressive model of Equation~\eqref{eq:autoreg1} by defining $\tXt\in\mbbN^{(M_1+F+1)\times (M_2+F+1)\times (M_3+F+1)}$  as 
$$
\tXtn = 
\begin{cases}
  \tYtmon & \text{if } n_q\leq M_q \ \forall q\\
  t^{F} & \text{if } n_q=M_q+F+1 \ \forall q\\
  0& \text{otherwise}
\end{cases}.
$$
Hence $\tXt$ has \{$\tYtmo,t^0,t^1,\dots,t^F$\} as subtensors. Under this $\tXt$, it holds that  
$$
\langle\tXt|\tB\rangle=\tM_0t^0+\tM_1t^1+\dots +\tM_Ft^F+ \langle\tYtmo|\widetilde\tB\rangle,
$$
where all $\{\tM_0,\tM_1\dots,\tM_F,\widetilde\tB\}$ are subtensors of $\tB$. Similarly, other kinds of trend (such as seasonal trend) can be incorporated by modifying $\tXt$ as above.

Long-term dependence can also be incorporated as part of the covariate $\tXt$. For instance, we can include a general $F$-th order autoregressive model by considering $\tXt\in\mbbN^{M_1\times M_2\times M_3\times F}$, where 
    $\tX_{(t)\mi{m,f}} = \tY_{(t-f)\mi{m}}$ ($f=1,\dots,F$). Hence, $(\tY_{(t-1)},\dots,\tY_{(t-F)})$ are all subtensors of $\tXt$, and from Equation~\eqref{eq:autoreg1}
    $$
    \langle\tXt|\tB\rangle=\langle\tYtmo|\tB_1\rangle+\langle\tY_{(t-2)}|\tB_2\rangle+\dots+\langle\tY_{(t-F)}|\tB_F\rangle,
    $$
    where $\{\tB_1,\tB_2,\dots\tB_F\}$ are all subtensors of $\tB$. 
    
\begin{figure}[ht]
\centering
\includegraphics[width=\linewidth]{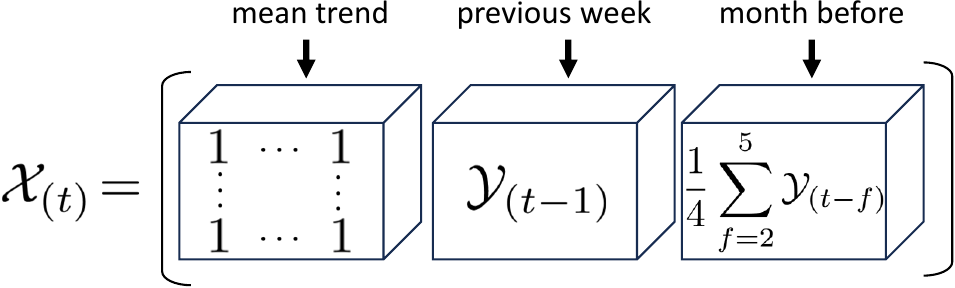}
\caption{4-way tensor covariates $\tXt$ used in the PToTR autoregressive model for the ICEWS experiments.}
\label{fig:longutidinalX}
\end{figure}
\subsubsection{Experimental results}
We compared PToTR, Gaussian ToTR~\citep{llosaandmaitra22}, and OP-based ToTR~\citep{hoff15} models involving the ICEWS database \citep{brien12}. 
We used the subset of the database described in \citep{hoff15}, which involves 25 countries and four quad classes as actions, leading to tensor responses $\tYt \in \mbbN^{25 \times 25 \times 4}$.  The dates selected comprise the 548 weeks between Thursday 01/01/2004 and Wednesday 07/02/2014. Hence, for $t=6, 7, \dots, T = 548$, our covariate tensor $\tXt \in \mbbR_{\geq 0}^{25 \times 25 \times 4 \times 3}$ comprises of the three tensors, each of dimension $25 \times 25 \times 4$, in a new fourth tensor dimension as illustrated in Figure \ref{fig:longutidinalX}. The three sub-tensors are $\bone_{25} \circ \bone_{25} \circ \bone_{4}$ (which integrates a mean trend), $\tYtmo$ (which integrates lag-1 dependence), and $\frac14\sum_{f=2}^5 \tY_{(t-f)}$ (which integrates longer-term dependence). Although we use the same data subset as in \citep{hoff15}, our preprocessing and modeling approach differs. Specifically, \citep{hoff15} applies a quantile–quantile transformation to align the data’s quantiles with those of the standard normal distribution. While this centers the data, the transformation is not invertible and thus discards information. In contrast, we work with the raw data and implicitly model the mean as part of the model.

Gaussian ToTR and OP-based ToTR were implemented using the \texttt{totr} R package of \citep{llosaandmaitra22} and PToTR was implemented using Algorithm \ref{alg:PToTR}. All models were initialized uniformly at random 100 times, with the chosen model having the largest resulting loglikelihood. 
We fitted the PToTR and Gaussian ToTR to the pairs of responses and covariates $(\tYt, \tXt)$ across different values of CP rank $R$ and display the resulting Bayesian Information Criterion (BIC) values \cite{schwarz78} in Figure \ref{fig:BICnRank}. For Gaussian ToTR, we tried many configurations of covariance matrices; however, none led to significantly different results from those presented in Figure \ref{fig:BICnRank}. Thus, we opted for the use of identity covariance matrices to ensure that PToTR and Gaussian ToTR result in the same number of parameters when the rank $R$ is the same. For comparison, we also fitted the OP-based ToTR of \citep{hoff15} to these data. Unlike the CP-based models, OP-based ToTR results in a different model for each mode permutation of $\tXt$. We fitted all 24 permutations and displayed the resulting BIC values as horizontal dashed lines in Figure \ref{fig:BICnRank}. 

\begin{figure}[ht]
\centering
\includegraphics[width=\linewidth]{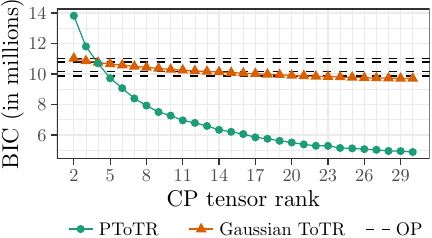}
\caption{BIC values from fitting PToTR, Gaussian ToTR, and OP-based ToTR models on the ICEWS database. PToTR outperforms its Gaussian counterpart for ranks greater than four, while Gaussian ToTR surpasses all OP-based ToTR for ranks exceeding 24. }
\label{fig:BICnRank}
\end{figure}

PToTR demonstrates superior performance compared to its Gaussian counterpart for $R > 4$ based on BIC values, indicating a better fit. The tensor $\tB \in \mbbRp^{25 \times 25 \times 4 \times 3 \times 25 \times 25 \times 4}$ contains over 18 million parameters without any low-rank constraints while each additional CP rank contributes only 111 additional parameters to the model.
This allows for the selection of a large rank $R$ while still achieving significant parameter reduction; for instance, a rank of $R = 30$ results in a 99.98\% reduction in parameters. In contrast, the number of parameters in the OP-based ToTR varies only with different mode permutations of the tensor $\tXt$, ranging from 297 to 1266. Notably, the best BIC values for the OP-based ToTR models corresponded to those with the highest number of parameters, suggesting that the OP-based ToTR may be too parsimonious to accurately represent the complex interactions contained in the tensor $\tB$. In summary, PToTR emerged as the best model, effectively modeling the random components through the Poisson distribution, and the complex relational interactions via the CP tensor model.

\subsection{Positron emission tomography imaging}\label{sec:pet}

Section~\ref{sec:motiv_pet} introduced the problem of  positron emission tomography (PET) image reconstruction and described the ML-EM method leveraging Poisson regression to solve this problem for 2-D data. In this section, we demonstrate that the PET reconstruction problem can be modeled using PToTR and show how it could be used in higher-order image reconstruction. We illustrate these extensions by simulating 4-D PET data from real imaging data and fit ML-EM and PToTR models to this data for comparison.

In PET scans, subjects are injected with a radiotracer containing a positron-emitting radionuclide whose beta decay produces a positron. This positron travels a few millimeters in tissue before encountering a nearby electron and annihilating, yielding two gamma photons that travel in nearly opposite directions. The PET scanner captures these gamma photons simultaneously within a brief time frame, allowing it to pinpoint the location of annihilation events in the body. This photon detection is recorded in a sinogram, which encodes the count of annihilations that occurred along photon paths at different angles and radial distances.
On the left display of Figure \ref{fig:phantom} we show the commonly-used Shepp-Logan phantom \citep{sheppandlogan74} as if it were under a PET scanner, with stars indicating the location of  three annihilations, and lines indicating the directions of the emitted photons. The right panel of Figure \ref{fig:phantom} indicates the corresponding position of these measurements within the observed sinogram. 

\begin{figure}[ht]
\centering
\includegraphics[width=\linewidth]{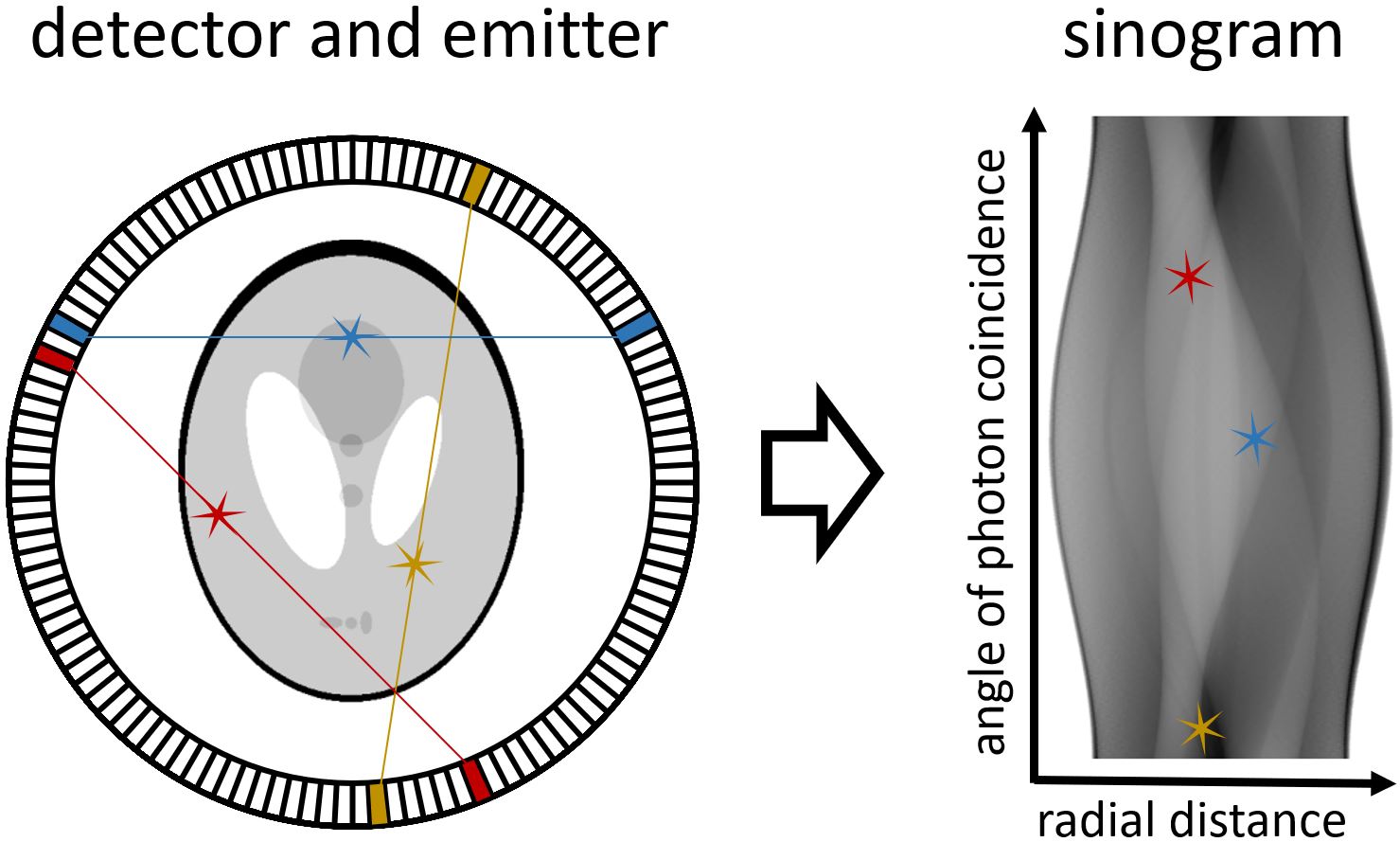}
\caption{The Shepp-Logan phantom under a PET scanner, with three annihilation events (colored stars) depicted on the left as they are being scanned and their corresponding locations in the resulting sinogram output depicted on the right.}
\label{fig:phantom}
\end{figure}

The goal of PET reconstruction is to estimate an image $\mat{B}$ (e.g., the Shepp-Logan phantom) under the scanner from the data captured in the observed sinogram $\mat{Y}$. 
A model for 2-D PET image reconstruction \citep{sheppandvardi82} can be written as
\begin{equation}\label{eq:PET1}
 \mat{Y} \sim \Poi( \mathcal{R}(\mat{B})) \; ,
\end{equation}
where $\mathcal{R}(\mat{B})$ denotes the discrete Radon transform of $\mat{B}$ and is of the same dimensions as $\mat{Y}$.  For example,  the phantom image $\mat{B}$ on the left of  Figure \ref{fig:phantom} has dimensions $(512,512)$ and the sinograms $\mat{Y}$ and $\mathcal{R}(\mat{B})$ have dimensions $(512,2048)$.

Next we introduce our PToTR-based PET image reconstruction method, which alleviates the ill-conditioned nature of multi-dimensional PET reconstruction by regularizing the image to have a low-rank CP tensor decomposition.
 
\subsubsection{Methodology}

Because the discrete Radon transform $\mathcal{R}(\mat{B})$ is a linear operator \citep{beylkin87}, we 
may write each entry as
\begin{equation}\label{eq:DRT}
\mathcal{R}(\mat{B})_{i_1,i_2} = \langle \mat{R}_{(i_1,i_2)},\mat{B}\rangle,    
\end{equation}
where $\mat{R}_{(i_1,i_2)}\in\mbbR^{N_1\times N_2}$ is a matrix of the same dimensions as $\mat{B}$. 
Above, we use the subscripts $(i_1,i_2)$ because these will correspond to PToTR observations. Based on an implementation of the discrete Radon transform, we can obtain the $\mi{n} = (n_1,n_2)$ entry of the matrix $\mat{R}_{(i_1,i_2)}$ using the identity
$$
\mat{R}_{(i_1,i_2)\mi{n}} = \mathcal{R}(\mat{E}_\mi{n})_{i_1,i_2},
$$
where $\mat{E}_\mi{n}\in\{0,1\}^{N_1\times N_2}$ is one at position $\mi{n} = (n_1,n_2)$ and zero elsewhere.  When the Radon transform is applied to $\mat{E}_\mi{n}$, we obtain the projection of that single point along various angles.
This will yield the sinogram $\mathcal{R}(\mat{E}_\mi{n})$, which represents how that point contributes to the overall image at different angles.  In practice, the matrices $\mat{R}_{(i_1,i_2)}$ also encode various machine-specific details, such as detector geometry, radiation source, and sampling resolution. 

Using the linearity of the discrete Radon transform in Equation \eqref{eq:DRT}, we can write the PET model of Equation \eqref{eq:PET1} as a scalar-response, matrix-covariate regression model 
\begin{equation}\label{eq:petmodel_vec}
    y_{(i_1,i_2)}\overset{\text{indep.}}{\sim} \Poi( \langle  \mat{R}_{(i_1,i_2)},\mat{B} \rangle),
\end{equation}
where the sample size is $I_1I_2$ ($i_1=1,\dots,I_1$ and $i_2=1,\dots,I_2$).
Equation \eqref{eq:petmodel_vec} is formulated for 2-D image reconstruction and can be extended to 3-D PET by stacking multiple 2-D images \citep{bendriemandtownsend13}.
 Similarly, stacking 3-D measurements across time form a 4-D PET scan, which has been suggested as a way to alleviate motion-related measurement degradation \citep{nehmehetal04,lietal06}. Our PToTR model allows us to extend the 2-D PET image reconstruction of Equation \eqref{eq:petmodel_vec} to $P$-D PET image reconstruction by stacking the scalar responses into a tensor:
\begin{equation}\label{eq:petmodel_tens}
    \tY_{(i_1,i_2)}\overset{\text{indep.}}{\sim} \Poi(  \langle \mat{R}_{(i_1,i_2)} | \tB\rangle),
\end{equation}
where $\tB\in \mbbRp^{N_1\times N_2\times M_1\times \dots\times M_{P-2}}$ is the $P$-dimensional image to be reconstructed and each $\tY_{(i_1,i_2)}\in \mbbN^{M_1\times \dots\times M_{P-2}}$ is a $(P-2)$-dimensional tensor response (e.g., a scalar in the 2-D PET model and a vector in the 3-D PET model). Equation \eqref{eq:petmodel_tens} is a Poisson-response tensor-on-matrix regression model that is a special case of PToTR in Equation \eqref{eq:ptotr}. 

While our $P$-D PET model is based on stacking multiple 2-D PET models, alternative definitions may employ different mathematical frameworks or data acquisition techniques. 
However, as long as the corresponding discrete Radon transform is a linear operator--i.e., has a form similar to Equation \eqref{eq:DRT}-- we can formulate the reconstruction problem as a specific instance of PToTR.
As a proof of concept, we next demonstrate our methodology on a synthetic 4-D PET experiment.

\subsubsection{Experimental results}

To illustrate the performance of PToTR for PET image reconstruction, we simulate data from the Poisson-response matrix-on-matrix regression model that is an instance of Equation \eqref{eq:petmodel_tens} when $P=4$. Here the tensor $\tB\in\mbbRp^{256\times 256\times 240\times 4}$ corresponds to four image measurements of a subject's brain of size $256\times 256\times 240$. 
 The real data in $\tB$ that we used is provided by \citet{hawcoetal22} and available for download from \url{https://openneuro.org} (accession number ds003011, version 1.2.3).
 For our simulation we chose measurements taken from the same subject (subject three), same location and scanner (Maryland Psychiatric Research Center using Siemens Tim Trio 3T), and during the second year of study.  
The matrix responses $\mat{Y}_\mi{(i_1,i_2)}\in \mbbN^{240\times 4}$ corresponds to the stacking of 4 longitudinal measurements across a depth dimension of size 240.  We use the Radon transform implementation of \citet{otnessandrim23}, transforming the $256\times 256$ ($N_1\times N_2$) axial images to $256\times 1024$  $(I_1\times I_2)$ sinograms.

We fitted the 4-D PET model of Equation \eqref{eq:petmodel_tens} using Algorithm \ref{alg:PToTR} with CP ranks $R=2,5,21,84,336$. While our regression model of Equation \eqref{eq:petmodel_tens} has a sample size of $I_1I_2=262,\!144$, in real scenarios only a subset of these would be observed; thus we performed estimation using randomly selected subsets containing 2\%, 4\%, 8\% and 16\% of the total data. 
For comparison, we also fitted the ML-EM algorithm that estimates $\tB$ without any restriction. For this estimation, one can iteratively use Theorem \ref{thm:nnstep} with
$\mat{Y} = [\vecc(\mat{Y}_{(1,1)})\dots \vecc(\mat{Y}_{(I_1,I_2)})]$ and  $\mat{D} = [\vecc(\mat{R}_{(1,1)})\dots\vecc(\mat{R}_{(I_1,I_2)})]\transpose$ until convergence. In this case, the resulting estimated matrix $\widehat{\mat{B}}$ is a matricization of the desired estimated tensor $\thB$, but without any rank constraint.

 In Figure \ref{fig:brainnorms}, we present the root mean square error (RMSE), defined as $\sqrt{||\thB - \tB||^2/d}$, where $d = 256 \times 256 \times 240 \times 4 = 62,914,560$ is the product of the dimension sizes $\tB$, across various reconstruction methods, sample sizes and number of iterations. The figure shows that for 2\% of the data, the RMSE increases with the number of iterations, indicating that the maximum likelihood estimate (MLE) is not an accurate reconstruction; as the algorithm approaches the MLE, the reconstruction becomes noisier in terms of RMSE. This trend persists for ML-EM as we increase the data percentage from 2\% to 16\%. In contrast, for PToTR($R$), we observe that increasing the rank $R$ or data percentage results in lower RMSE as the number of iterations increases; specifically, for 16\% of the data, PToTR fits using all ranks exhibit monotonically decreasing RMSE. This suggests that, unlike ML-EM, the closer the estimate gets to the MLE, the better the reconstruction becomes in terms of RMSE. For the 4\%, 8\%, and 16\% cases, ML-EM shows a sharp decrease in RMSE during the initial iterations, followed by an increase as iterations continue. While this initial low RMSE is comparable to that of PToTR, it is important to note that one typically cannot determine when this optimal RMSE will occur without knowledge of the ground truth. In contrast, with an appropriate sample size and/or rank, our PToTR reconstructions improve with more iterations, regardless of whether $\tB$ is known. Furthermore, unlike ML-EM, which treats each element in $\tB$ as a parameter to be estimated (resulting in $d$ parameters), our PToTR approach is significantly more parsimonious. For example, PToTR with rank $R = 84$ has only 63,168 parameters, making it nearly three orders of magnitude more parsimonious than ML-EM.

\begin{figure}[t!]
\centering
\includegraphics[width=\linewidth]{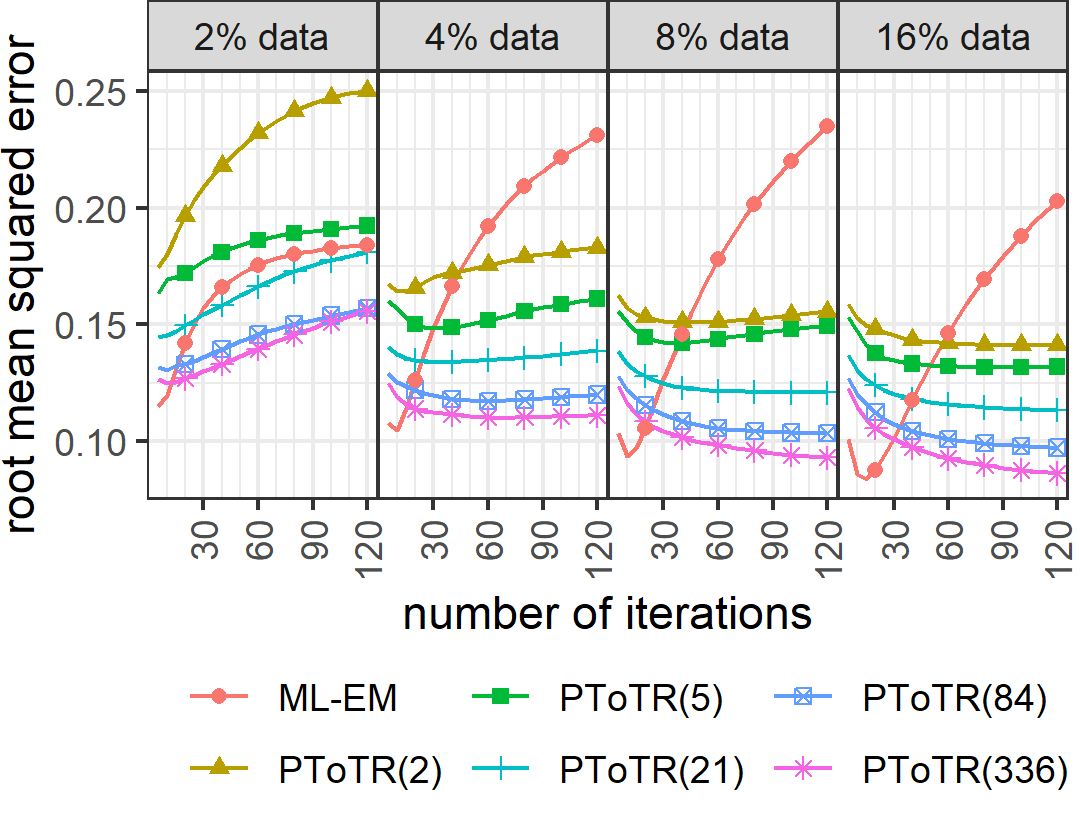}
\caption{RMSE ($\propto\sqrt{||\thB - \tB||^2}$) across different reconstruction methods, sample sizes, and iterations in the estimation. Results show that for increased data and parameters, PToTR($R$), improves with increased iterations, while ML-EM exhibits an initial decrease in RMSE followed by a substantial increase.
}
\label{fig:brainnorms}
\end{figure}

Figure \ref{fig:figbrains} shows example ground truth images from $\tB$ across three different planes: an axial image from the first frame ($\tB_{:,:,120,1}$), a coronal image from the second frame ($\tB_{:,128,:,2}$), and a saggital image from the third frame ($\tB_{128,:,:,3}$). The figure also shows image reconstructions using ML-EM and PToTR with ranks 84 and 336, modeled using 4\% and 16\% of the data, after 10 and 120 iterations. Consistent with previous results, we observe that a larger number of iterations improves the reconstruction quality for PToTR, while resulting in noisier reconstructions for ML-EM. Visually, it appears that PToTR with rank $R=336$ and for 120 iterations, retains most of the gray matter details across all planes and percentages of data.
 
 \begin{figure*}[ht]
\centering
\includegraphics[width=\linewidth]{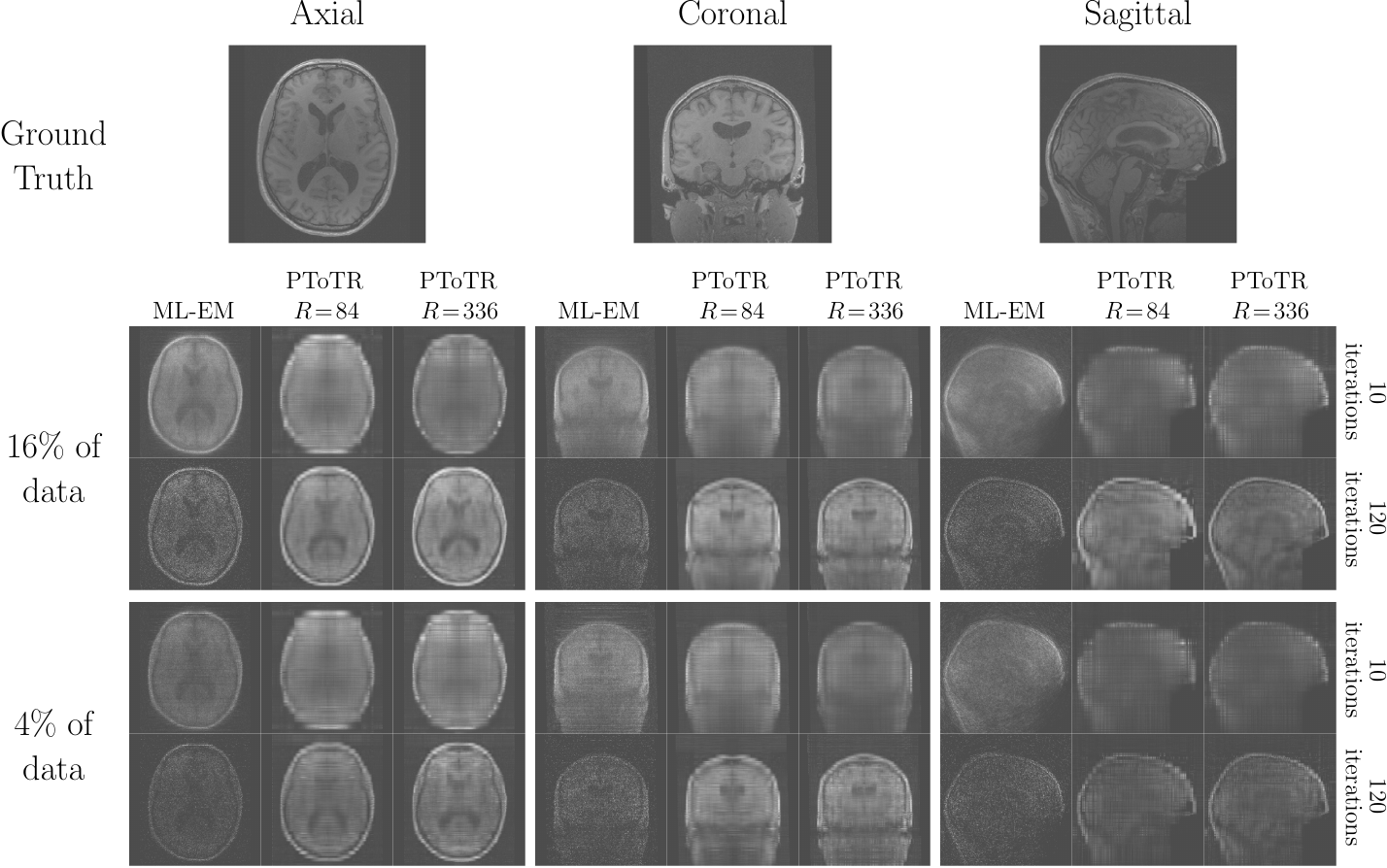}
\caption{PET reconstructions for different amounts of data (4\%, 16\%), different number of iterations (10 and 100), and multiple reconstruction methods (ML-EM, and PToTR for ranks 84 and 336). Unlike ML-EM, the recovery of our PToTR method is increased for larger numbers of iterations.}
\label{fig:figbrains}
\end{figure*}

\subsection{Change-point detection in dyadic data}\label{sec:ptanova}
Section \ref{sec:motiv_changepoint} introduced the  problem of change-point detection for tensor responses containing count data. One-way analysis of variance (ANOVA) can be used for change-point detection by distinguishing means before and after some point in a sequence of data~\citep{AlAl14}. Tensor-variate ANOVA (TANOVA) extends ANOVA to the general case of tensor responses and was introduces as a special case of ToTR using indicator covariates~\cite{llosaandmaitra22}. In the Gaussian case, TANOVA has been used to detect brain regions with significant interaction between death-related stimuli and suicide attempt status, and also to distinguish facial characteristics across different factors \citep{llosaandmaitra22}. In this section, we introduce Poisson-response TANOVA (PTANOVA) as a special case of PToTR and demonstrate its use in change-point detection on simulated communication data over time, where responses are tensors of counts.

\subsubsection{Methodology}
Consider communications between a group of $M_2$ senders and $M_3$ (possibly same) receivers interacting about $M_1$ different topics over $T$ discrete times. Similar to Section \ref{sec:autoregressive} we can encode the communication that occur at a particular time $t$ as a tensor $\tYt\in\mbbN^{M_1\times M_2\times M_3}$, with $\tYtm$ entry corresponding to the number of times a communication involving topic $m_1$ was sent from $m_2$ to $m_3$ during time $t$. Suppose an event occurs at time $\tau\in\{1,2,\dots,T-1\}$ that causes the communication patterns to change going forward in time. If we assume that the communications vary according to a Poisson distribution, this change in topic can be modeled as
\begin{equation}\label{eq:changepoint}
\tYt \sim 
\Bigg\{
    \begin{array}{lr}
        \Poi(\tB^{(1)})  &t =1,2,\dots,\tau\hfill\\
        \Poi(\tB^{(2)})  & t=\tau+1,\tau+2,\dots,T
    \end{array}.
\end{equation}
In Equation \eqref{eq:changepoint} we assume that the communications follow a Poisson distribution with mean $\tB^{(1)}$ before the event and a different mean $\tB^{(2)}$ after the event. The goal is to estimate $\tB^{(1)}$, $\tB^{(2)}$, and $\tau$. 
 Equation \eqref{eq:changepoint} can be equivalently written in PTANOVA form as
\begin{equation}\label{eq:tanovachangepoint}
    \tYt \sim \Poi(\langle \bx_t | \tB\rangle),
\end{equation}
where the covariates are $\bx_t = [1,0]\transpose$ for $t=1,2,\dots,\tau$, $\bx_t = [0, 1]\transpose$ for $t=\tau+1,\tau+2\dots,T$, and the regression coefficient $\tB \in \mbbRp^{2\times M_1\times M_2\times M_3}$ contains the two Poisson coefficients $\tB^{(1)} =\langle [1, 0]\transpose | \tB\rangle$ and $\tB^{(2)} =\langle [0, 1]\transpose | \tB\rangle$ as subtensors.
While equations \eqref{eq:changepoint} and \eqref{eq:tanovachangepoint} are equivalent, the PTANOVA formulation in Equation \eqref{eq:tanovachangepoint} is framed in terms of the  single regression coefficient tensor $\tB$, assumed to have CP structure
\begin{equation}
    \tB = [\![ \blambda; \mVq{1},\mUp{1},\dots,\mUp{P}]\!].
\end{equation}

Our formulation of change-point detection in Equations \eqref{eq:changepoint} and \eqref{eq:tanovachangepoint} depends on a known value of change-point $\tau$. We estimate $\tau$ through maximum likelihood estimation. First, denote the loglikelihood that results from fitting the PToTR of Equation \eqref{eq:tanovachangepoint} for a fixed value of $\tau$ as $\ell_{\tau}(\thB_{\tau})$. Then
\begin{equation}\label{eq:ptotr_esttau}
\widehat{\tau} = \argmax_{\tau}\ell_{\tau}(\thB_{\tau})
\end{equation}
corresponds to the maximum likelihood estimate for $\tau$. This estimate has a few other interpretations.  Note that since all  $\ell_{\tau}(\thB_{\tau})$ have the same number of parameters, choosing $\widehat\tau$ as in Equation \eqref{eq:ptotr_esttau} is equivalent to choosing the model with the smallest BIC or AIC. Furthermore, consider the hypotheses 
$$
H_0: \tB^{(1)} = \tB^{(2)} 
,\quad 
H_A: \tB^{(1)} \neq \tB^{(2)}.
$$
A likelihood ratio test statistic for these hypothesis is 
$$
\Lambda_{\tau} = 2(\ell_{\tau}(\thB_{\tau}) - \ell_{0}(\thB_{0})),
$$
where $\ell_{0}(\thB_{0})$ is the loglikelihood that corresponds to the case where $\tB^{(1)} = \tB^{(2)}$, and it can be obtained from fitting a PToTR with responses $\tYt$ and covariates $1$ for all $t=1,2,\dots,T$. The test statistic $\Lambda_{\tau}$ is largest when $\ell_{\tau}(\thB_{\tau})$ is largest. Hence, our estimated $\widehat\tau$ from Equation \eqref{eq:ptotr_esttau} corresponds to the model with the largest evidence against the null hypothesis of no change-point  \citep{killickandeckley14}, \citep{granjon97}.

The matrix factors can be estimated according to Algorithm \ref{alg:PToTR}, which is greatly simplified after considering the vector-variate and binary nature of the covariates.  For these simplifications first let
$
\tY_{1\tau} = \sum_{t=1}^\tau \tYt, 
\tY_{2\tau} = \sum_{t=\tau+1}^{T} \tYt,
\mVq{1}  = [\bv^{(1)}_1 \ \bv^{(1)}_2]\transpose,
\mat{U}=\odot_{p} \mUp{p}$, and 
$\mat{U}_{-p}=\odot_{k\neq p} \mUp{k}$.
Then Equation \eqref{eq:MLestV} is simplified for $(\tau_1,\tau_2) = (\tau,T-\tau)$ as
\begin{equation*}\label{eq:simpMLestV2}
\mtVq{1}_{\{k+1\}}  \leftarrow \mtVq{1}_{\{k\}}   * 
\begin{bmatrix} 
\left[ \vecc\left(\tY_{1\tau}\right)\oslash\left(\mat{U} \widetilde\bv_{1\{k\}}^{(1)}  \right)  \right]\transpose\mat{U}/\tau_1
\\
\left[ \vecc\left(\tY_{2\tau}\right)\oslash\left(\mat{U} \widetilde\bv_{2\{k\}}^{(1)} \right)  \right]\transpose\mat{U}/\tau_2
\end{bmatrix}.
\end{equation*}
Similarly, Equation \eqref{eq:MLestU} is simplified as
\begin{equation*}\label{eq:simpMLestB2}
\begin{aligned}
\mtUp{p}_{\{k+1\}} \leftarrow
\mtUp{p}_{\{k\}}  &* 
\Bigg\{
\sum_{j=1}^2
\left[\left(
[\tY_{j\tau}]_{(p)}\oslash(\mtUp{p}_{\{k\}}\mat{G}_{jp})
\right)\mat{G}_{jp}\transpose\right]
\Bigg\}\\
&\oslash
\left\{
\bone\bw\transpose
\right\},
\end{aligned}
\end{equation*}
where $\bw = \tau_1 \bv^{(1)}_1   + \tau_2 \bv^{(1)}_2  $, $\mat{G}_{jp} = \mat{U}_{-p}\diag(\bv^{(1)}_j  )$.

\subsubsection{Experimental results}

In this experiment we simulate a setting where 10 subjects communicate with each other across 15 topics and 14 time-steps, so each $\tYt\in \mbbN^{10\times 10\times 15}$ and $t=1,2,\dots,14$.  Each $\tYt$ is generated element-wise from a Poisson distribution with rate $\omega=1$ except for one of the 15 topics after time $\tau$, which uses rate $\omega=a$.  This means that the Poisson rate for one of the $10\times 10$ matrix slices of $\tYt$ is $a$ whenever $t>\tau$ and is $1$ when $t\leq \tau$. For this experiment we chose $a\in\{2,4,8\}$, and the true value of $\tau=\in\{0,1,4,6\}$. Note that $\tau = 0$ means that there is no change-point and $\tau = 1$ means that there was only one time instance before the change-point. 
We applied the PTANOVA model described in Equation \eqref{eq:tanovachangepoint} for each value of $a$, considering all candidate change-points $\tau = 1, 2, \ldots, 13$, and for CP ranks of $R = 2, 4, 6, 8$. The resulting loglikelihoods from each model fit are presented in Figure \ref{fig:llks}, where vertical dashed lines represent the change-point values in each plot.
\begin{figure}
\includegraphics[width=\linewidth]{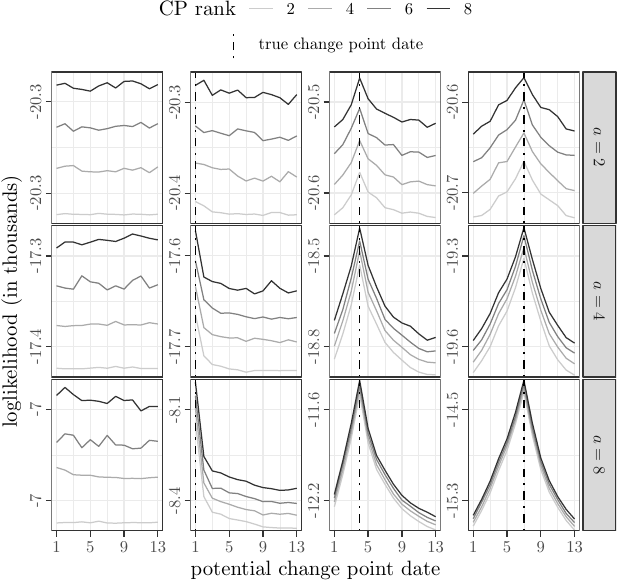}
\caption{Loglikelihoods resulting from fitting the PTANOVA model of Equation \eqref{eq:tanovachangepoint} for different values of $a$, $\tau$, CP rank. There is a clear peak at the true location of change-point for all the cases that have a change-point (right three columns). There is no clear peak for the case with no change-point (left column).
}\label{fig:llks}
\end{figure}

In Figure \ref{fig:llks}, we observe distinct peaks in loglikelihood at most of the true change-point locations, indicating successful identification of the change-point in almost all applicable cases. The only exception occurs at $a = 2$ and $\tau = 1$, which is anticipated since there is only one time point preceding the change-point, and this change-point reflects only doubling in communication volume. The loglikelihoods all peak at the true change-point value for $\tau = 1$ when $a = 4$ and $a = 8$, as well as for $a = 2$ when $\tau = 4$ and $\tau = 6$. This pattern is consistent throughout the figure: clearer peaks in loglikelihood are associated with larger values of $a$ and change-points $\tau$ that are closer to the center of the time period of communications. Conversely, the left column, which represents scenarios without a change-point, shows no distinct peaks in loglikelihood. Furthermore, larger ranks $R$ yield higher loglikelihoods, as expected. Importantly, we successfully detected the change-point across all values of $R$ except when $a = 2$ and $\tau = 1$. In conclusion, our analysis demonstrates that the PTANOVA model effectively identifies change-points, with improved detection for more drastic changes and a greater number of observed groups before and after the change-point. This suggests the model is robust in various settings, reinforcing its utility in analyzing communication patterns over time.

\section{Conclusions and future work}\label{sec:conclusion}
Poisson-response tensor-on-tensor regression (PToTR) represents a significant advancement in the field of multi-dimensional data analysis, offering a principled and versatile framework that integrates the strengths of Poisson tensor decompositions and tensor-on-tensor regression (ToTR). By leveraging the statistical properties of Poisson-distributed data and the relational modeling capabilities of ToTR, PToTR addresses the unique challenges posed by complex, structured count data. This innovative approach is both theoretically sound and practically effective as demonstrated through its diverse applications in predicting political events, reconstructing PET images, and detecting change points in dyadic data.

There are several directions for future work. One potential avenue is the incorporation of a log link function in the Poisson regression model. The log link function ensures that the expected counts are always positive and can model multiplicative relationships between responses and covariates (instead of additive relationships as we have done here), providing a different framework for certain types of count data. Another promising direction is the extension of PToTR to a generalized ToTR (GToTR) model that could allow for general response distributions (e.g., binomial or negative binomial) and link functions (e.g., log, logit, probit, etc.), thereby broadening the applicability of the method to a wider range of data types and data analysis problems. Additionally, exploring other low-rank tensor models for the regression coefficient tensor, such as the Tucker and tensor train (TT) decompositions, could offer further benefits. While these are promising avenues for future work, it is important to note that the effectiveness of each approach, including our current PToTR, will depend on the characteristics and requirements of the data in each case.

\section*{Acknowledgments}

We thank Carolyn D. Mayer, J. Derek Tucker, and Jonathan Berry of Sandia National Laboratories for several helpful discussions during the writing of this manuscript. We also thank Ranjan Maitra of The Department of Statistics at Iowa State University for his useful insights on the PET model.

This paper describes objective technical results and analysis. Any subjective views or opinions that might be expressed in the paper do not necessarily represent the views of the U.S. Department of Energy or the United States Government.

\bibliographystyle{IEEEtran} 
\bibliography{references} 

\end{document}